\documentclass{article}
\usepackage[a4paper]{geometry}
\usepackage{amssymb,amsmath,amsthm}
\usepackage{graphics,graphicx,color}
\usepackage{microtype,verbatim,dsfont}
\usepackage{xspace,enumerate}
\usepackage{subcaption}
\usepackage{todonotes}
\usepackage[pdfpagelabels,colorlinks,citecolor=blue,linkcolor=blue,urlcolor=blue]{hyperref}
\usepackage[capitalise]{cleveref}
\usepackage[normalem]{ulem}
\usepackage{enumitem}
\usepackage{apxproof}
 
\theoremstyle{plain}
\newtheorem{theorem}{Theorem}
\newtheorem{definition}{Definition}
\newtheorem{lemma}[theorem]{Lemma}
\newtheorem{corollary}[theorem]{Corollary}

\newtheorem{property}[theorem]{Property}
\newtheorem{claim}[theorem]{Claim}

\Crefname{observation}{Observation}{Observations}
\Crefname{algorithm}{Algorithm}{Algorithms}
\Crefname{algocf}{Algorithm}{Algorithms}
\Crefname{section}{Section}{Sections}
\Crefname{lemma}{Lemma}{Lemmata}
\Crefname{claim}{Claim}{Claims}
\Crefname{property}{Property}{Properties}
\Crefname{enumi}{Property}{Properties}
\Crefname{figure}{Fig.}{Figs.}

\newcommand{\myparagraph}[1]{\medskip\noindent\textbf{#1}}

\DeclareMathOperator{\sk}{sk}
\newcommand{\qn}{\mathrm{qn}}
\newcommand{\tw}{\mathrm{tw}}
\newcommand{\simple}[1]{\ensuremath{\mathrm{si}(#1)}}

\newcommand{\floor}[1]{\lfloor #1 \rfloor}
\newcommand{\ca}[1]{{\mathcal{#1}}}
\def\Ksize{3\floor{h/2}+\floor{h/3}-1}

\title{Graph Product Structure for $h$-Framed Graphs}

\author{Michael A. Bekos$^1$, Giordano Da~Lozzo$^2$, Petr~Hliněný$^3$, Michael~Kaufmann$^4$\\
\medskip
\\
\small$^1$University of Ioannina, Ioannina, Greece\\
\small\texttt{bekos@uoi.gr}
\\
\small$^2$Department of Engineering, Roma Tre University, Italy\\
\small\texttt{giordano.dalozzo@uniroma3.it}
\\
\small$^3$Masaryk University, Brno, Czech Republic\\
\small\texttt{hlineny@fi.muni.cz}
\\
\small$^4$University of T{\"u}bingen, T{\"u}bingen, Germany\\
\small\texttt{mk@informatik.uni-tuebingen.de}
}
\date{}

\begin{document}

\maketitle

\begin{abstract}
Graph product structure theory expresses certain graphs as subgraphs of the strong product of much simpler graphs. In particular, an elegant formulation for the corresponding structural theorems involves the strong product of a path and of a bounded treewidth graph, and allows to lift combinatorial results for bounded treewidth graphs to graph classes for which the product structure holds, such as to planar graphs [Dujmović et al., J.\ ACM, 67(4), 22:1-38, 2020].

In this paper, we join the search for extensions of this powerful tool beyond planarity by considering the $h$-framed graphs, a graph class that includes $1$-planar, optimal $2$-planar, and $k$-map graphs (for appropriate values of $h$). We establish a graph product structure theorem for $h$-framed graphs stating that the graphs in this class are subgraphs of the strong product of a path, of a planar graph of treewidth at most $3$, and of a clique of size $\Ksize$. This allows us to improve over the previous structural theorems for $1$-planar and $k$-map graphs. Our results constitute significant progress over the previous bounds on the queue number, non-repetitive chromatic number, and $p$-centered chromatic number of these graph classes, e.g., we lower the currently best upper bound on the queue number of $1$-planar graphs and $k$-map graphs from $495$ to $81$ and from $32225k(k-3)$ to $61k$, respectively. We also employ the product structure machinery to improve the current upper bounds of twin-width of planar and $1$-planar graphs from $183$ to $37$, and from $O(1)$ to~$80$, respectively. All our structural results are constructive and yield efficient algorithms to obtain the corresponding~decompositions.
\end{abstract}

\section{Introduction}
\label{sec:introduction}

Graph product structure theory~\cite{DBLP:journals/jacm/DujmovicJMMUW20} was recently introduced and is receiving~considerable attention, as it gives deep insights that allow a host of mathematical and algorithmic tools to be applied. 
Despite being a relatively new~development, it is  having significant impact~\cite{Banff2021}; initially, it was introduced to settle a long-standing conjecture by Heath, Leighton and Rosenberg~\cite{DBLP:journals/siamdm/HeathLR92} related to the queue number of planar~graphs~\cite{DBLP:journals/jacm/DujmovicJMMUW20}, but recently it has been exploited to solve several other combinatorial problems that were~open for years, e.g., it was used to prove that planar graphs have bounded non-repetitive chromatic number~\cite{DBLP:journals/jacm/DujmovicJMMUW20}, to improve the best known bounds for $p$-centered colorings of planar graphs and graphs excluding any fixed graph as a subdivision~\cite{DBLP:conf/soda/DebskiFMS20}, to find shorter adjacency labelings of planar graphs~\cite{DBLP:conf/soda/BonamyGP20}, and to find asymptotically optimal adjacency labelings of planar graphs~\cite{DBLP:conf/focs/DujmovicEGJMM20}. 

In its simplest form, the product structure theorem states that every planar graph is a subgraph of the strong product of a path and of a planar graph of treewidth at most 6~\cite{DBLP:journals/jacm/DujmovicJMMUW20,DBLP:journals/corr/abs-2108-00198}. The bound on the treewidth can be improved by allowing more than two graphs in the strong product, as it is known that every planar graph is a subgraph of the strong product of a path, of a $3$-cycle and of a planar graph of treewidth at most $3$~\cite{DBLP:journals/jacm/DujmovicJMMUW20}. These theorems are attractive, since they describe planar graphs in terms of graphs of bounded treewidth, which are considered much simpler than the planar ones. Furthermore, they enable combinatorial results that hold for graphs of bounded treewidth to be generalised for planar graphs and, more in general, for graphs where similar structural theorems can be obtained.

Analogous are known for graphs of bounded Euler genus~\cite{DBLP:journals/jacm/DujmovicJMMUW20},  apex-minor-free graphs~\cite{DBLP:journals/jacm/DujmovicJMMUW20},  graphs with bounded degree in minor closed classes~\cite{DBLP:journals/cpc/DujmovicEMWW22}, and  graphs in non-minor closed classes~\cite{DBLP:journals/corr/abs-1907-05168}; see~\cite{DBLP:journals/corr/abs-2001-08860} for a survey. 
Related to our work are the structural theorems for $k$-planar and $k$-map graphs 
(the former ones are the graphs that can be drawn with at most $k$ crossings per edge, whereas the latter ones are the contact-graphs of regions homeomorphic to closed disks). In particular, it is known that every $k$-planar graph is a subgraph of the strong product of a path, of a graph of treewidth at most $\frac{1}{6}(k+4)(k+3)(k+2)-1$, and of a clique on $18k^2+48k+30$ vertices, while every $k$-map graph is a subgraph of the strong product of a path, of a  graph of treewidth at most $9$, and of a clique on $21k(k-3)$ vertices~\cite{DBLP:journals/corr/abs-1907-05168}.

\myparagraph{Our contribution.} 
In this work, our focus is on the class of $h$-framed graphs, which were recently introduced as a notable subclass of $k$-planar and $k$-map graphs (for appropriate values of $k$)~\cite{DBLP:conf/compgeom/BekosLGGMR20}; a graph is \emph{$h$-framed}, if it admits a drawing on the Euclidean plane whose uncrossed edges induce a biconnected spanning plane graph with faces of size at most $h$. Since any $h$-framed graph is $O(h^2)$-planar, it follows that every $h$-framed graph is a subgraph of a path, of a graph with treewidth $O(h^6)$ and of a clique of size $O(h^4)$. Our main contribution is to show the following structural result (which lowers the treewidth to $O(1)$ --also achieving planarity-- and the size of the involved clique to $O(h)$): Every $h$-framed graph is a subgraph of the strong product of a path, of a planar graph of treewidth at most 3, and of a clique on $\Ksize$ vertices; see \cref{thm:prodstruct-h}. Note that, since any planar graph is a subgraph of some triangulation (and thus of a $3$-framed graph), for~$h=3$ we have that planar graphs are subgraphs of $H \boxtimes P \boxtimes K_{3}$, which coincides with the product structure theorem for planar graph proved in~\cite{DBLP:journals/jacm/DujmovicJMMUW20}. Furthermore, we provide an alternative formulation, where the role played by a path is instead played by the $\floor{h/2}$-th power of a path, which allows us to further reduce the size of the clique involved in the product to $\max(3,h-2)$; see \cref{thm:prodstruct-magic}. Our algorithms provide improved upper bounds on the queue number, on the non-repetitive chromatic number, and on the $p$-centered chromatic number of $h$-framed graphs that are linear in $h$ (while the ones that can be derived from the result in~\cite{DBLP:journals/corr/abs-1907-05168} are at least quadratic in $h$); see~\cref{th:queue-h-framed}, \cref{thm:non-repetitive}, and \cref{lem:centered}, respectively. Finally, by extending the product structure machinery, we are able to give an efficient construction to obtain an improved bound on the twin-width of planar graphs and an explicit, linear in $h$, upper bound on the twin-width of $h$-framed graphs, while the current-best explicit upper bound derives from the one for $k$-planar graphs and it is hence  exponential~in~$O(h^2)$~\cite{DBLP:journals/jacm/BonnetKTW22,DBLP:journals/corr/abs-2202-11858}; see~\cref{cref:ttwin-width-planar}~and~\cref{cref:ttwin-width-hframed}~respectively. 

\myparagraph{Consequences on related graph classes.}
Since $1$-planar and optimal $2$-planar graphs are subgraphs of $4$- and $5$-framed graphs~\cite{DBLP:conf/gd/AlamBK13,DBLP:conf/compgeom/Bekos0R17}, and since $k$-map graphs are subgraphs of $2k$-framed graphs~\cite{DBLP:journals/corr/abs-2003-07655,DBLP:conf/compgeom/BekosLGGMR20}, the product structure theorems mentioned above imply significant improvements on the current best bounds for the following parameters; for definitions,~see~\Cref{se:consequences}.

\begin{itemize}
\item \textbf{Queue number:} Using \cref{thm:prodstruct-h}, we improve the best-known upper bound on the queue number of $k$-map graphs from $32225k(k-3)$~\cite{DBLP:journals/corr/abs-1907-05168} to $61k$ (\cref{thm:queuenumber-kmap}),~whereas, using \cref{thm:prodstruct-magic}, we lower the best-known upper bounds on the queue number of $1$-planar and optimal $2$-planar graphs from $495$ and $3267$, respectively, both to $81$ (\cref{thm:queuenumber-1p-o2p}). 
\item \textbf{Non-repetitive chromatic number:} \cref{thm:prodstruct-h} allows us to improve the best-known upper bound on the non-repetitive chromatic number of $k$-map graphs from \mbox{$21 \cdot 4^{10} \cdot k(k-3)$}~\cite{DBLP:journals/corr/abs-1907-05168} to $256(3k + \floor{2k/3} - 1)$, whereas for the class of $1$-planar graphs our improvement is from $30 \cdot 4^4$ to $6 \cdot 4^4$, which is a bound that notably holds also for optimal $2$-planar graphs (\cref{thm:non-repetitive}). 
\item \textbf{$p$-centered chromatic number:} \cref{thm:prodstruct-h} allows us to improve the best-known upper bound on the non-repetitive chromatic number of $k$-map graphs from $O(k^2 p^{10})$~\cite{DBLP:journals/corr/abs-1907-05168} to $O(k  p^3 \log p)$, whereas for the class of $1$-planar graphs our improvement is from $O(p^4)$ to $O(p^3 \log p)$, and this bound also holds for optimal $2$-planar graphs~(\cref{thm:centered}). 
\item \textbf{Twin-width:} \cref{cref:ttwin-width-planar} improves the current-best upper bound on the twin-width of planar graphs from $183$~\cite{DBLP:journals/corr/abs-2201-09749} to $37$. For the class of $1$-planar and optimal $2$-planar graphs, \cref{cref:ttwin-width-hframed} improves the bound from $O(1)$~\cite{DBLP:journals/jacm/BonnetKTW22} to $80$, whereas our improvement for $k$-map~graphs is limited to certain value of $k$, as these graphs have bounded twin-width~independently~of~$k$~\cite{DBLP:journals/jacm/BonnetKTW22}.
\end{itemize}

\section{Preliminaries}
\label{sec:preliminaries}

\begin{figure}[tb!]
	\centering
	\begin{subfigure}[b]{.48\textwidth}
    	\centering
    	\includegraphics[page=4]{figures/1-planar.pdf}
    	\caption{}
    	\label{fig:1-planar-skeleton}
	\end{subfigure}
	\hfil
	\begin{subfigure}[b]{.48\textwidth}
    	\centering
    	\includegraphics[page=1]{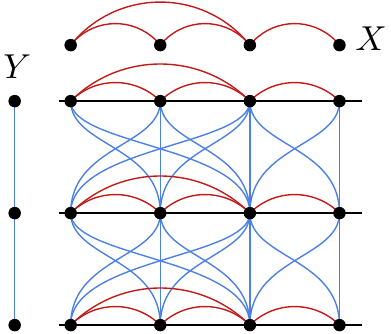}
    	\caption{}
	\label{fig:strong}
    \end{subfigure}
	\caption{Illustration of: (a) a $4$-framed topological graph whose skeleton edges (crossing edges) are black (blue), and
	(b) the strong product $X \boxtimes Y$ of a planar graph $X$ (red) and a path~$Y$~(blue).}
\end{figure}

For standard graph-theoretic terminology and notation we refer the reader, e.g., to~\cite{DBLP:books/daglib/0030488}.

\myparagraph{Graphs.} 
A graph is \emph{simple} if it contains neither loops nor multi-edges. For a general graph~$G$ (not-necessarily simple), let $\simple{G}$ denote the \emph{simplification} of $G$, i.e., the simple graph obtained from $G$ by removing all loops and replacing each bunch of parallel edges with a single edge. 
For any $i\geq 1$, the \emph{$i$-th power} $G^i$ of an graph $G$ is the graph with the same vertex set as $G$, in which two vertices are adjacent if and only if they are at distance at most~$i$ in~$G$. Clearly, $G \subseteq G^i$. A graph $H$ is a \emph{minor} of a graph $G$, if $H$ can be obtained from a subgraph of $G$ by contracting~edges.

\myparagraph{Topological graphs.}
A \emph{topological graph} is a graph drawn on the plane such that any two edges cross in at most one point and no edge crosses itself. In this paper, we will solely consider topological graphs in which no two adjacent edges cross and no three edges cross in the same point 
(in the literature, such drawings are commonly referred to as ``simple'').
A \emph{plane graph} is a topological graph with no crossing edges.
A graph is \emph{$k$-planar} if it is isomorphic to a topological graph in which each edge crosses at most $k$ other edges. Furthermore, a $k$-planar graph of maximum density is called \emph{optimal}. A \emph{$k$-map graph} is one that admits a \emph{$k$-map}, i.e., a contact representation on the sphere, where each vertex is a region isomorphic to a closed disk and no more than $k$ regions share the same boundary~point.
Given a topological graph $G$, the subgraph $\sk(G)$ of a topological graph $G$ consisting of all its vertices and uncrossed edges is the \emph{skeleton} of $G$; refer to \cref{fig:1-planar-skeleton}. A topological graph $G$ whose skeleton $\sk(G)$ is biconnected is called \emph{$h$-framed}~\cite{DBLP:conf/compgeom/BekosLGGMR20}, if all the faces of $\sk(G)$ have size at most~$h$, and \emph{internally $h$-framed}, if all the faces of $\sk(G)$, except for possibly one, have size at most~$h$. The importance of this class lies in the following connections with $k$-planar and $k$-map graphs~\cite{DBLP:journals/corr/abs-2003-07655,DBLP:conf/compgeom/BekosLGGMR20}. Optimal $1$-planar and optimal $2$-planar graphs are $4$- and $5$-framed, respectively, while general $1$-planar graphs can be augmented to $8$-framed graphs, if multi-edges are forbidden, or to $4$-framed graphs, if multi-edges are allowed. Finally, note that any $k$-map graph is a subgraph of a $2k$-framed graph. 

\myparagraph{Treewidth.}
Let $(\mathcal{X},T)$ be a pair such that $\mathcal{X}=\{X_1,X_2,\dots,X_\ell\}$ is a collection of subsets of vertices of a graph $G$, called \emph{bags}, and $T$ is a tree whose nodes are in one-to-one correspondence with the elements of $\mathcal X$. The pair $(\mathcal{X},T)$ is a \emph{tree-decomposition} of $G$ if it satisfies the following two conditions: (i)~for every edge $(u,v)$ of $G$, there exists a bag $X_i \in \mathcal{X}$ that contains both~$u$ and~$v$, and (ii)~for every vertex $v$ of $G$, the set of nodes of $T$ whose bags contain $v$ induces a non-empty subtree of $T$. The \emph{width} of a tree-decomposition $(\mathcal{X},T)$ of $G$ is $\max_{i=1}^\ell {|X_i|} - 1$, while the \emph{treewidth} $\tw(G)$ of $G$ is the minimum width over all tree-decomposition of $G$. 

\myparagraph{Quotient graph.}
For a graph $G$ and a partition  $\ca P$ of $V(G)$, 
the \emph{quotient of $G$ by~$\ca P$}, denoted by $G/\ca P$, is a graph containing a vertex $v_P$ for each part $P$ in $\mathcal P$ (we say that $v_P$ \emph{stems~from}~$P$) and an edge $(v_{P'}, v_{P''})$ if and only if there exists a vertex in $P'$ adjacent to a vertex in~$P''$ in $G$. Note that, $G/\ca P$ is a minor
of~$G$, if every part in $\ca P$ induces a connected subgraph of $G$. 

\myparagraph{Strong product.}
The \emph{strong product} of two graphs $X$ and $Y$, denoted by $X \boxtimes Y$, is the graph whose vertex-set $V(X \boxtimes Y)$ is the Cartesian product $V(X) \times V(Y)$, such that there exists an edge in $E(X \boxtimes Y)$ between the vertices $\langle x_1,y_1 \rangle, \langle x_2,y_2 \rangle \in V(X \boxtimes Y)$ if and only if one of the following occurs: (a)~$x_1 = x_2$ and $(y_1,y_2) \in E(Y)$, (b)~$y_1 = y_2$ and $(x_1,x_2) \in E(X)$, or (c)~$(x_1,x_2) \in E(X)$ and $(y_1,y_2) \in E(Y)$; see \cref{fig:strong}. Dujmovic et al.~\cite{DBLP:journals/jacm/DujmovicJMMUW20,DBLP:journals/corr/abs-1907-05168} and Ueckerdt, Wood, and Yi~\cite{DBLP:journals/corr/abs-2108-00198} showed the following main graph product structure results.

\begin{theorem}[Dujmovic et al.~\cite{DBLP:journals/jacm/DujmovicJMMUW20,DBLP:journals/corr/abs-1907-05168}, Ueckerdt et al.~\cite{DBLP:journals/corr/abs-2108-00198}]\label{th:k-map-planar}
For a graph $G$, the next hold:
\begin{enumerate}[label={\alph*.}, ref=(\alph*)]
\item\label{th:planar-1} If $G$ is planar, then $G \subseteq P \boxtimes H$, for a path $P$ and a planar graph $H$ with~$\tw(H) \leq 6$. 
\item\label{th:planar-2} If $G$ is planar, then $G \subseteq P \boxtimes H \boxtimes K_{3}$, for a path $P$ and a planar graph $H$ with~$\tw(H) \leq 3$. 
\item\label{th:one} If $G$ is $1$-planar, then $G \subseteq P \boxtimes H \boxtimes K_{30}$, for a path $P$ and a planar \mbox{graph $H$ with $\tw(H) \leq 3$}.
\item\label{th:k} If $G$ is $k$-planar with $k>1$, then $G \subseteq  P \boxtimes H \boxtimes K_{18k^2+48k+30}$, for a path $P$ and a graph $H$ with $\tw(H) \leq \frac{1}{6}(k+4)(k+3)(k+2)-1$.
\item \label{th:map} If $G$ is a $k$-map graph, then $G \subseteq P \boxtimes H \boxtimes K_{21k(k-3)}$, for a path $P$ and a graph~$H$ with $\tw(H) \leq 9$.
\end{enumerate}
\end{theorem}

\myparagraph{Layering.}
Consider a graph $G$.
A \emph{layering} of $G$ is an ordered partition $(V_0,V_1,\dots)$
of $V(G)$ such that, for every edge $(v,w)$ of $G$ with $v\in V_i$ and $w\in V_j$, it holds $|i-j|\leq 1$. If $i=j$, then $(v,w)$ is an \emph{intra-level edge}; otherwise, $(v,w)$ is an \emph{inter-level edge}.
Each part~$V_i$ is called a \emph{layer}.
Let $T$ be a BFS-tree of $G$ rooted at a vertex $r$.
The \emph{BFS layering} of $G$ \emph{determined by~$T$} is the layering  $(V_0,V_1,\ldots)$ of $G$ such that $V_i$ contains all vertices of $G$ at distance $i$ from $r$. Given a partition $\mathcal P$ of $V(G)$ and a layering $\mathcal L$ of $G$, the \emph{layered width~of~$\mathcal P$ with respect to $\mathcal L$} is the size of the largest set obtained by intersecting a part in $\mathcal P$ and a layer in $\mathcal L$. The \emph{layered width} of $\mathcal P$ is the minimum layered width of $\cal P$ over all layerings of $G$.

\section{Computing the Product Structure}
\label{sec:getstructure} 
\noindent This section is devoted to the proof of a product structure theorem for $h$-framed graphs, summarized in the next theorem; several applications of this result are~presented~in~\cref{se:consequences}.

\begin{theorem}[Product Structure Theorem for $h$-Framed Graphs]\label{thm:prodstruct-h}
	Let $G$ be a not-necessarily simple $h$-framed graph with $h \geq 3$.
	Then, $\simple{G}$ is a subgraph of the strong product $H \boxtimes P \boxtimes K_{\Ksize}$, where $H$ is a planar graph with $\tw(H) \leq 3$ and $P$ is a path.
\end{theorem}

\noindent
The algorithm supporting \cref{thm:prodstruct-h} is going to recursively decompose the graph $G$ into parts with special properties, such that the resulting quotient graph will be $H$, and
the additional properties of the constructed partition will imply the claimed product structure. We start with a technical setup followed by the core recursion in \cref{lem:recurse}.

\myparagraph{Layering~$\mathbf{G}$.}
Let $T$ be a BFS tree of $\sk(G)$ rooted at an arbitrary vertex $r$ incident to the outer face of $\sk(G)$. For an arbitrary $H$ and its implicitly fixed BFS tree $T'\subseteq H$ (such as $H=\sk(G)$ and $T'=T$ in our case), we call a path $P\subseteq\sk(G)$ \emph{vertical} if $P$ is a subpath of some root-to-leaf path of~$T'$. Let $\mathcal L = (V_0,V_1,\ldots,V_{b})$ be the BFS layering of $\sk(G)$ determined by~$T$. Observe that, if $P$ is a vertical path in $\sk(G)$, then $P$ intersects every part of $\ca L$ in at most one vertex. Given~$\mathcal L$, we define a new ordered partition $\mathcal W= (W_0,W_1,\dots,W_\ell)$ of the vertex set of $G$ with $\ell=\big\lceil b/\floor{\frac{h}2}\big\rceil-1$, by merging consecutive $\lfloor\frac{h}2\rfloor$-tuples of layers of~$\ca L$. This is done as follows. For $i=0,1,\ldots,\ell$, we let $W_i:=\bigcup_{j=0}^{\floor{h/2}-1}V_{i\floor{h/2}+j}$ (assuming $V_{x}=\emptyset$ if $x>b$). Then, $\ca W:=(W_0,W_1,\ldots,W_{\ell})$ is a layering of $G$, as we formally prove below.

\begin{property}\label{prop:valid-layering}
$\ca W:=(W_0,W_1,\ldots,W_{\ell})$ is a layering of $G$.
\end{property}
\begin{proof}
As for the edges of $\sk(G)$, we have that intra-level edges of $\cal L$ are also intra-level edges of $\cal W$, whereas inter-level edges of $\cal L$ are either intra-level edges or inter-level edges of~$\cal W$. Next, we argue about the crossing edges of $G$. First, observe that each such an edge is a chord in some face of $\sk(G)$. Also, every chord of a face $f$ of $\sk(G)$ has its ends at distance at most $\floor{\frac{h}2}$ along $f$. This implies that $G$ does not contain an edge with $(u,v)$ with $u \in W_i$ and $v \in W_j$ with $|i-j|>1$, which in turn implies that $\mathcal W$ is a layering of $G$. 
\end{proof}

\myparagraph{Partitioning~$\mathbf{G}$.}
The core of our algorithm is a construction of a special partition $\ca R$ of $V(G)$ such that $H = G/\ca R$ is a planar graph with $\tw(H) \leq 3$, and the layered width of~$\ca R$ with respect to~$\ca W$ is not large. Our recursive decomposition of~$G$ is analogous to the one in~\cite{DBLP:journals/jacm/DujmovicJMMUW20} (as applied to planar graphs); however several non-trivial changes are needed to exploit the existence of the underlying (plane) skeleton of $G$. The algorithm starts from the outer face and recursively ``dive'' into gradually-shrinking areas of~$G$.
	
Central in our approach is the following notion. For a cycle $C\subseteq \sk(G)$, the \emph{subgraph of~$G$ bounded by~$C$}, denoted by $G_C$, is the subgraph of $G$ formed by the vertices and edges of $C$ and the vertices and edges of $G$ drawn inside~$C$. Consider a subset $U\subseteq V(G)$. For the partition $\ca L$ (resp.,\ the partition $\ca W$), the \emph{width of $U$ with respect to $\ca L$} (resp.\ \emph{to $\ca W$}), denoted by $\lambda_{\ca L}(U)$ (resp.\ by $\lambda_{\ca W}(U)$), is the largest size of a set obtained by intersecting $U$ and a part of~$\mathcal L$ (resp.\ of $\ca W$). We are now ready to present our main technical lemma. 
	
\begin{lemma}\label{lem:recurse}
Let $G$ be an $h$-framed graph with $h \geq 3$ and let $\ca L$ be a BFS-layering of $G$. Also, let $C$ be a cycle in~$\sk(G)$, and let $G_C$ be the subgraph of $G$ bounded by~$C$. Further, for some $k\in\{1,2,3\}$, let $P_1,\ldots,P_k$ be paths belonging to $C$ such that $\mathcal R^0 = \{ X_i: X_i: = V(P_i), 1 \leq i \leq k\}$ is a partition~of~$V(C)$. Then, it is possible to construct in quadratic time a \emph{good} partition $\ca R'$ of~$V(G_C)$, i.e., one that satisfies the following properties:
\begin{enumerate}
\item \label{cl:inclusion} $\mathcal R' \supseteq \ca R^0$, and for every part $X \in \mathcal R'\setminus\ca R^0$, there exist $q\in\{1,2,3\}$ and $X'\subseteq X$ such that%
\footnote{The somehow technical \Cref{cl:inclusion} of \cref{lem:recurse} will imply that $\lambda_{\ca W}(X) \leq \Ksize$ in the proof of \cref{thm:prodstruct-h}, but we will also make use of the stated more detailed treatment.}
\begin{itemize}
    \item $X\setminus X'$ is a union of the vertex sets of at most $q$ vertical paths of $\sk(G)$, and so, in particular, $\lambda_{\ca L}(X\setminus X') \leq q$, and
    \item $|X'|\leq h-3$ if $q=1$,~ $|X'|\leq \floor{(h-1)/2}-1$ if $q=2$, and $|X'|\leq \floor{h/3}-1$ if~$q=3$.
\end{itemize}
\item\label{cl:II-treewidth} the quotient graph $H ' = G_C/\ca R'$ is a planar graph with $\tw(H') \leq 3$, and
\item\label{cl:II-triangle} the vertices of $H'$ that stem from $X_i$, with $1 \leq i \leq k$, are incident to the same face of $H'$ and induce a clique (i.e., either a vertex, or an edge, or a triangle).
\end{enumerate}
\end{lemma}

\begin{proof}
We prove \cref{lem:recurse} by providing a recursive procedure describe in the following. In the base case of the recursion occurs when $V(G_C)=V(C)$ (i.e., there are no vertices in the interior of~$C$ {and the edges in $E(G_C)\setminus E(C)$ are chords of $C$}). In this case, the algorithm returns the partition $\ca R'=\ca R^0$, which is clearly good since the graph $H'$ is a plane clique of size $k$ whose vertices stem from the parts of $\ca R^0$. 	Note that, if $|E(G_C)\setminus E(C)|=0$, then $V(G_C)=V(C)$, since $G_C$ cannot have isolated vertices.
In the recursive step of the algorithm, we assume that there exist vertices and edges of~$G_C$ that lie in the interior of $C$. Our aim is to recurse on instances that contain fewer edges in the interior, but not on the boundary, of the cycle bounding their outer face. We first need to handle a possible degenerate case\footnote{Such a case does not explicitly occur in the planar proof of~\cite{DBLP:journals/jacm/DujmovicJMMUW20}, but the implicit case of a so-called~\cite{DBLP:journals/jacm/DujmovicJMMUW20} ``tripod'' with degenerate legs is analogous to what we are defining here.} of~$G_C$. Recall that, since $G$ is $h$-framed, all bounded faces of $G_C\cap \sk(G)$ have length at most~$h$.
	
\begin{figure}[tb!]
	\centering
	\begin{subfigure}{0.3\textwidth}
	\centering
	\includegraphics[page=1]{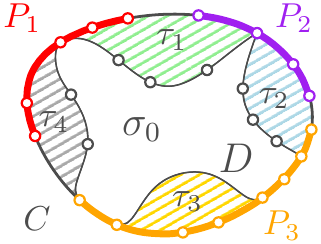}
	\caption{}
	\label{fig:separable-many}
	\end{subfigure}
	\hfill
	\begin{subfigure}{0.3\textwidth}
	\centering
	\includegraphics[page=5]{figures/separable.pdf}
	\caption{}
	\label{fig:separable-one}
	\end{subfigure}
	\hfill
	\begin{subfigure}{0.3\textwidth}
	\centering
	\includegraphics[page=2]{figures/separable.pdf}
	\caption{}
	\label{fig:separable-no}
	\end{subfigure}
	\caption{Illustrations of graph $C \cup D$ for the separable case of $G_C$
	(\Cref{def:separable}). The number $a$ of $\tau$-faces is $4$ in (a) and $1$ in (b). (c) A case when \cref{def:separable-at-most-two} of \cref{def:separable} is~not~met~(by~$\tau_4$).
	\label{fig:separable}
	}
\end{figure}
	
\begin{definition}\label{def:separable}
We say that $G_C$ is \emph{separable} if the following conditions hold (see \cref{fig:separable}):
\begin{enumerate}
\item \label{def:separable-subfaces} The plane graph $G_C\cap \sk(G)$ contains a bounded face $\sigma_0$ that intersects $C$ in $a\geq1$ disjoint subpaths (some of the paths may consist of single vertices). Denoting by $D$ the facial cycle of $\sigma_0$, let $\tau_1,\ldots,\tau_a$ be the bounded faces of the plane graph $C\cup D$ other~than~$\sigma_0$.
\item \label{def:separable-at-most-two} For each $\tau_j$, with $j = 1,\ldots,a$, the boundary of $\tau_j$ is a cycle~$C_j$ of~$\sk(G)$, at most two parts of $\ca R^0$ intersect $C_j$, and every part of $\ca R^0$ intersecting $C_j$ does so in a single subpath.\footnote{A part of $\ca R^0$ indeed may intersect the boundary~$C_j$ of $\tau_j$ in two subpaths (see \cref{fig:separable-no}). This, however, can happen only if~$k\leq2$.}
\end{enumerate}
\end{definition}

\myparagraph{Separable case.} 
Suppose that $G_C$ is separable and let $G_j$, $j=1,\ldots,a$, denote the subgraph of $G_C$ bounded by the facial cycle $C_j$ of~$\tau_j$. Note that, if $a=1$, we have $|V(D)\cap V(C)|\geq2$, since otherwise the face $\tau_1$ would not be bounded by a cycle as required by~\cref{def:separable}. This implies that $|E(G_j)\setminus E(C_j)|<|E(G_C)\setminus E(C)|$ (even when~$a=1$). Also, let~$Y$ denote the vertices of $D$ that do not belong to $C$, i.e., $Y=V(D)\setminus V(C)$; refer to the hollow vertices of \cref{fig:separable}. By the previous, we have $|Y|\leq|D|-2\leq h-2$.
For $j = 1,\ldots,a$, let $\ca R^0_j$ be the partition of $V(C_j)$ consisting of the set $Y_j = Y\cap V(C_j)$, if it is not empty, and of the sets $X^j_{i} = X_i \cap V(C_j)$, $i = 1,\dots,k$, if $X^j_{i}$ is not empty; by \cref{def:separable-at-most-two} of \cref{def:separable}, $R^0_j$ consists of at most three parts. Therefore, by a recursive application of our algorithm, each graph $G_j$, $j=1,\ldots,a$, admits a good partition $\mathcal R_j \supseteq \ca R^0_j$~of~$V(G_j)$.
We construct a partition $\ca R'$ of $V(G_C)$ by putting into $\ca R'$ the parts of $\ca R_0$, the set $Y$ (if non-empty), and the recursively obtained parts of each $G_j$ that do not touch $C_j$; formally, $\ca R'=\ca R_0\cup \{Y\}\cup \bigcup_{j=1,\ldots,a}(\ca R_j\setminus\ca R^0_j)$, or $\ca R'=\ca R_0\cup \bigcup_{j=1,\ldots,a}(\ca R_j\setminus\ca R^0_j)$ if $Y=\emptyset$. Note that $\ca R'$ is indeed a partition of $V(G_C)$, since each vertex of $G_C$ that lies in the interior of $C$ must belong either to $Y$ or to a part $X \in \ca R_j \setminus \ca R^0_j$, for some $j \in \{1,\dots,a\}$. To prove that $\ca R'$ is good, we need a preliminary property given below.

\begin{property}\label{cl:tripod} 
Under the conditions of \cref{lem:recurse}, no vertex of any part from $\ca R'\setminus\ca R^0$ is adjacent to a vertex of $V(G)\setminus V(G_C)$.
\end{property}
\begin{proof}
Since $\ca R^0$ partitions $V(C)$, and $C$ is a cycle in $\sk(G)$, the drawing of $C$ is an uncrossed simple closed curve in the topological graph $G$, and hence, by the Jordan curve theorem, no edge of $G$ can have one end in $V(G_C)\setminus V(C)$ and the other end in $V(G)\setminus V(G_C)$.
\end{proof}

\noindent We are now ready to show that the constructed partition $\ca R'$ is good.

\begin{claim}\label{claim:good-separable}
The partition $\ca R'$ constructed for the separable case of $G_C$ is good.
\end{claim}
\begin{proof}
For our proof, we use \Cref{cl:tripod} mentioned above, which we can further assume that it holds also for the recursive calls. First, we look at \cref{cl:inclusion} of \cref{lem:recurse}: this property holds true for every $X\in\ca R'\setminus(\ca R_0\cup \{Y\})$ by the recursive construction, and for $X=Y$ we pick any $y\in Y$ and set $Y'=Y\setminus\{y\}$, and by the previous we get $|Y'|\leq|Y|-1\leq h-2-1=h-3$ and $Y\setminus Y'=\{y\}$ is a trivial vertical path. 

It remains to analyze planarity and treewidth of the quotient graph $H'$. Recall that we have recursively obtained the planar quotient graphs $H_j:=G_j/\ca R_j$, for $j \in \{1,\dots,a\}$, and we may assume, from the recursive invocation of \Cref{cl:II-triangle}, that each $H_j$ is a plane topological graph with the vertices stemming from the parts of $\ca R^0_j$ on the outer face. For further reference, we call these vertices stemming from $\ca R^0_j$ the \emph{connectors} of~$H_j$. By \Cref{cl:tripod} following a recursive invocation of \cref{lem:recurse}, no vertices of $H_j$ other than the connectors will be adjacent to vertices of $H'$ outside of $H_j$.

We start with drawing a plane $(k+1)$-clique $Q$ on the vertices stemming from the non-empty parts in $\ca R^0\cup\{Y\}$. If $k=2$, then at most two of the graphs $G_j$ with $j \in \{1,\dots,a\}$ intersect both parts of $\ca R^0$, and for them we embed the corresponding (at most two) plane graphs $H_j$ into the two triangular faces of~$Q$, such that the vertices of $Q$ are naturally identified with the connectors of~$H_j$. The remaining quotient graphs $H_j$ are then embedded into the drawing easily, since the connectors of each of them is identified with only one or two vertices of~$Q$. Likewise, the desired plane drawing of $H'$ is trivial if~$k=1$. If $k=3$, then each pair of parts of $\ca R^0$ is intersected together by at most one of the graphs $G_j$, and then we embed the plane quotient graph $H_j$ into the corresponding triangular face of~$Q$, with the appropriate identification of the connectors of $H_j$ as in the case of~$k=2$. Again, the remaining quotient graphs $H_j$ are embedded easily.

Altogether, we have obtained a plane drawing of $H'$ such that the vertices stemming from the parts of $\ca R^0$ are on the same face; in particular, for $k=3$, this is the triangular face of $Q$ not containing the vertex which stems from the part $Y$. Moreover, the $k$ parts of $\ca R^0$ are indeed pairwise adjacent in $G_C$ already by edges of~$C$. We have thus verified \Cref{cl:II-triangle} of \cref{lem:recurse}, and it remains to verify \Cref{cl:II-treewidth} in the aspect of treewidth of $H'$. We have recursively obtained, for each $j \in \{1,\dots,a\}$, a tree-decomposition $\ca T_j$ of the graph $H_j$, such that (by a folklore property of tree-decompositions) the clique of the vertices which stem from $\ca R_j^0$ is contained in a node $\nu_j$ of it. We create a new decomposition $\ca T$ of $H'$ from the disjoint union of $\ca T_j$ over $j=1,\ldots,a$ by adding a new node $\nu$, holding the bag of vertices $V(Q)$ and adjacent exactly to all $\nu_j$. Further, for $i=1,\dots,k$, we rename each vertex of $\mathcal{T}_j$ that stems from $X^j_i$, with $i=1,\dots,a$, as the vertex in $Q$ that stems from $X_i$ and, for $i=1,\dots,a$, we rename each vertex of $\mathcal{T}_j$ that stems from $Y_i$ as the vertex in $Q$ that stems from $Y$. This is a valid tree-decomposition by  \Cref{cl:tripod}, and it is of width at most $3$ by  \Cref{cl:II-treewidth} and the fact that $|V(Q)|-1=k\leq3$.
\end{proof}

\myparagraph{General case.} 
\begin{figure}[tb!]
	\centering
	\begin{subfigure}{0.32\textwidth}
	\centering
	\includegraphics[page=6,scale=1]{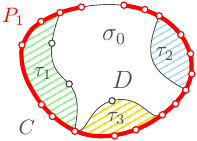}
	\caption{}
	\label{fig:general-1}
	\end{subfigure}
	\hfil
	\begin{subfigure}{0.32\textwidth}
	\centering
	\includegraphics[page=8,scale=1]{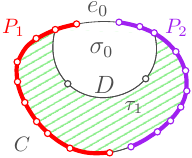}
	\caption{}
	\label{fig:general-0}
	\end{subfigure}
	\hfil
	\begin{subfigure}{0.32\textwidth}
	\centering
	\includegraphics[page=7,scale=1]{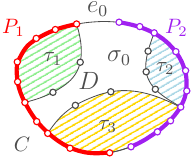}
	\caption{}
	\label{fig:general-2}
	\end{subfigure}
	\caption{Illustrations of graph $C \cup D$ for the general case of $G_C$, for $k=1$ (a) and $k=2$ (b)(c).}
	\label{fig:general-1-2}
\end{figure}
Now we move to the general (i.e., not-necessarily separable) case of $G_C$ in \Cref{lem:recurse}. If $\mathbf{k=1}$, we pick the bounded face $\sigma_0$ of $G_C\cap\sk(G)$ incident to the single edge of $E(C)\setminus E(P_1)$; refer to \cref{fig:general-1}. The face $\sigma_0$ then witnesses the separable case for $G_C$, by~\cref{def:separable}, which is solved as above. If $\mathbf{k=2}$, then we pick $e_0\in E(C)$ as one of the edges joining $P_1$ and $P_2$ on~$C$, and $\sigma_0$ as the bounded face of $G_C\cap\sk(G)$ incident to~$e_0$;  see \cref{fig:general-0,fig:general-2}. Then we are back to the separable case for $G_C$ with~$\sigma_0$, by~\cref{def:separable}.

In the remainder, we assume $\mathbf{k=3}$. First, we color every vertex $v$ of $G_C$ by the color $i\in\{1,2,3\}$ if the (unique) path in the BFS tree $T$ from $v$ to the root~$r$ hits $V(P_i)$ before possibly hitting other parts of~$\ca R^0$. In particular, the vertices of $P_i$ are colored~$i$. Our aim is to find, in the plane graph $F:=G_C\cap \sk(G)$, a bounded face $\sigma_1$ containing vertices of all the three colors on its boundary. There our arguments divert from those used in~\cite{DBLP:journals/jacm/DujmovicJMMUW20}---since $F$ is generally not a near-triangulation, and we additionally need that the face $\sigma_1$ intersects the boundary cycle~$C$ at most once (which requires additional care). We next prove that there exists a cycle $R$ bounding a bounded face $\sigma_1$ of $F$, such that $V(R)$ contains all three of our colors, and $R$ intersects $C$ in at most one connected piece. Furthermore, the colors on $R$ appear in three consecutive sections. To show this claim, we need the next definition.

\begin{figure}
    \centering
    \begin{subfigure}[b]{.24\textwidth}
        \centering
        \includegraphics[page=1,width=\textwidth]{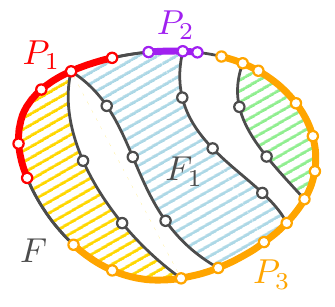}
        \subcaption{}
    	\label{fig:zones-F}
    \end{subfigure}	
    \hfil
        \begin{subfigure}[b]{.24\textwidth}
        \centering
        \includegraphics[page=4,width=\textwidth]{figures/zones.pdf}
        \subcaption{}
    	\label{fig:zones-general-F3}
    \end{subfigure}	
    \hfil
    \begin{subfigure}[b]{.24\textwidth}
        \centering
        \includegraphics[page=6,width=\textwidth]{figures/zones.pdf}
        \subcaption{}
    	\label{fig:zones-general-sigma3}
    \end{subfigure}	
    \hfil
    \begin{subfigure}[b]{.24\textwidth}
        \centering
        \includegraphics[page=5,width=\textwidth]{figures/zones.pdf}
        \subcaption{}
    	\label{fig:zones-general-sigma1}
    \end{subfigure}	
    \caption{Illustrations for:
    (a)~Graph $F$ and a zone $F_1$.
    (b)~Graph $F_3$.
    (c)~The tri-colored face $\sigma_3$ in $F_3$.
    (d)~The desired face $\sigma_1$ in $F$.
    }
    \label{fig:zones}
\end{figure}

\begin{definition}
A plane subgraph $F_1\subseteq F$ is a \emph{zone of $F$} if the following hold (see \cref{fig:zones-F}):
\begin{itemize}
\item every vertex in $V(F_1)\cap V(C)$ is of degree at least~$2$ in~$F_1$,
\item no vertex in $V(F)\setminus V(F_1)$ is adjacent to a vertex in $V(F_1)\setminus V(C)$, and
\item in the plane drawing of $F_1$ inherited from~$G_C$, every vertex from $V(F)\setminus V(F_1)$ is drawn in the outer face of~$F_1$.
\end{itemize}
\end{definition}
	
Note that, similarly to \cref{def:separable}, if a facial cycle $D$ in $F$ intersects $C$ in two or more places, then each of the remaining bounded faces of the plane subgraph $C\cup D$ bounds a zone of~$F$. We say that the face of $D$ \emph{divides} $F$ into these zones. 

For our proof, we pick $F_1$ as an inclusion-minimal zone of $F$ such that $V(F_1)$ intersects all three parts of~$\ca R^0$ (i.e., $F_1$ meets all three colors on~$C$).
We claim that we can assume that there is no bounded face $\sigma_0$ of $F_1$ such that $\sigma_0$ intersects $C$ in at least two disjoint subpaths. On the contrary, suppose that such $\sigma_0$ exists in $F_1$. Then $F_1$ can be divided into smaller zones which, by the minimality of $F_1$, each intersects at most two parts of~$\ca R^0$. However, the same then holds for the zones into which $\sigma_0$ divides the whole~$F$, and so $\sigma_0$ witnesses the case of separable~$G_C$ already solved, as can be easily checked from~\cref{def:separable}.
We take the graph $C\cup F_1$ and temporarily replace (i.e., ``shortcut'') each maximal path of $C-E(F_1)$ by a single new edge. Let~$F_2$ be the resulting graph (so,~$V(F_1)=V(F_2)$), and $F_3$ be an arbitrary completion of $F_2$ to a plane near-triangulation on the same vertex set, i.e., obtained by adding edges into the bounded faces of~$F_2$; refer to \cref{fig:zones-general-F3}. Now, in the inherited coloring of~$V(F_3)\subseteq V(G_C)$, the outer face of $F_3$ sees all three colors in a way expected by Sperner's Lemma\footnote{We adopt the following well-known variant of Sperner's Lemma: Let $G$ be a near-triangulation whose vertices are colored with three colors, where the vertices of the cycle $C$ bounding the outer face of $G$ belonging to each color class form a single subpath of $C$. Then $G$ contains an internal face whose vertices have all the three colors~\cite{DBLP:books/daglib/0019107}.}, and by this lemma we can thus find a tri-colored triangular face $\sigma_3$ of $F_3$; refer to \cref{fig:zones-general-sigma3}. By planarity of $F$ and $F_3$, there is a unique face $\sigma_1$ of $F$ containing $\sigma_3$; refer \cref{fig:zones-general-sigma1}. Let $R\subseteq V(F)$ denote the facial cycle of~$\sigma_1$ ($R$ is a cycle since the skeleton $\sk(G)$ is $2$-connected). We have that $V(R)$ contains all three of our colors, and $R$ intersects $C$ in at most one connected piece by our choice of~$F_1$. Furthermore, the colors on $R$ appear in three consecutive sections since the paths of the BFS tree $T$ do not cross in the plane graph~$F$. This concludes the proof of our initial claim.

\begin{figure}[tb!]
    \centering
    \begin{subfigure}[b]{.32\textwidth}
	    \centering
	    \includegraphics[page=3,width=\textwidth]{figures/tripods.pdf}
	    \subcaption{Case~\ref{it:vrvc-case-a}}
    	\label{fig:vrvc-case-a}
    \end{subfigure}	
    \hfil
    \begin{subfigure}[b]{.32\textwidth}
	    \centering
	    \includegraphics[page=4,width=\textwidth]{figures/tripods.pdf}
	    \subcaption{Case~\ref{it:vrvc-case-b}}
    	\label{fig:vrvc-case-b}
    \end{subfigure}	

    \begin{subfigure}[b]{.32\textwidth}
	    \centering
	    \includegraphics[page=2,width=\textwidth]{figures/tripods.pdf}
	    \subcaption{Case~\ref{it:vrvc-case-c}}
    	\label{fig:vrvc-case-c}
    \end{subfigure}	
    \hfil
    \begin{subfigure}[b]{.32\textwidth}
	    \centering
	    \includegraphics[page=1,width=\textwidth]{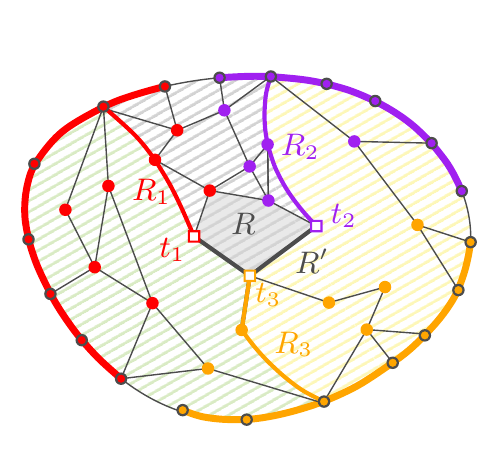}
	    \subcaption{Case~\ref{it:vrvc-case-c}}
    	\label{fig:vrvc-case-cc}
    \end{subfigure}	
    \caption{Illustrations for the representatives $t_1$, $t_2$, and $t_3$, and the vertical paths $R_1$, $R_2$,~and~$R_3$. The vertices of $R$ bound the gray shaded region. $R'$ is depicted with black thick edges.
    }
    \label{fig:vrvc-cases}
\end{figure}

Next, consider the set $V(R)\cap V(C)$.  If this set contains all three colors, then all three colors occur on the path $R_0:=C\cap R$, and one of them, say color $1$, occurs only on internal vertices of $R_0$ (and nowhere else on $C$). In this case, the face $\sigma_1$ again witnesses the case of separable~$G_C$ with~$a=1$ which is solved as above. If, instead, the set $V(R)\cap V(C)$ does not contain all three colors, then we choose on $R$ representatives -- vertices $t_i\in V(R)$ of color~$i$ where $i=1,2,3$, as in one of the following three possible cases of $V(R)\cap V(C)$ (refer to \cref{fig:vrvc-cases,fig:vrvc-case-cc}):
\begin{enumerate}[label={\bf C.\arabic*}, ref=C.\arabic*]
\item \label{it:vrvc-case-a} If $V(R)\cap V(C)$ contains two colors, say $1$ and $2$, we choose $t_1,t_2\in V(R)\cap V(C)$ as neighbors on $C$ and $t_3\in V(R)\setminus V(C)$ arbitrarily; refer to \cref{fig:vrvc-case-a}. 

\item \label{it:vrvc-case-b} If $V(R)\cap V(C)$ contains one color, say $1$, we choose $t_1\in V(R)\cap V(C)$ arbitrarily and pick $t_2,t_3\in V(R)\setminus V(C)$ such that $t_2t_3\in E(R)$ (this is unique). Furthermore, up to symmetry between the colors $2$ and $3$, we may assume that the distance on $R$ between $t_2$ and $V(C)$ is not smaller than the distance on $R$ between $t_3$ and~$V(C)$; refer to \cref{fig:vrvc-case-b}.

\item \label{it:vrvc-case-c} If $V(R)\cap V(C)=\emptyset$, then, up to symmetry between the colors, we may assume that the color $3$ occurs in $V(R)$ no more times than each of the colors $1$ and~$2$. We then choose $t_3\in V(R)$ arbitrarily (of color~$3$), and set $t_1$ and $t_2$ to the two (unique) vertices colored $1$ and $2$ on $R$ that are neighbors of vertices of color $3$ on~$R$; refer to \cref{fig:vrvc-case-c,fig:vrvc-case-cc}.
\end{enumerate}

For $i=1,2,3$, let $R_i$ denote the unique vertical path in $T$ from $t_i$ to $V(P_i)$; see~\cref{fig:vrvc-cases,fig:vrvc-case-cc}. Note that some vertices $t_i$ may lie on $C$, and then $R_i=t_i$ is a single-vertex path. Let~$Q$ be the subpath of $R$ with the ends $t_1$ and $t_2$ and avoiding~$t_3$. We define $R'\subseteq R$ as the subpath or cycle (in the case~$R'=R$) obtained from $R$ by deleting all internal vertices of~$Q$. Finally, we set $R^+:=R'\cup R_1\cup R_2\cup R_3$ which is a connected subgraph of $F$ ($R^+$ will play the same role here as the so-called tripods in~\cite{DBLP:journals/jacm/DujmovicJMMUW20}). 

Observe that $C\cup R^+$ is a $2$-connected plane graph (in each of the three cases above), which has $a\in\{2,3\}$ bounded faces $\tau_1,\ldots,\tau_a$ that are moreover distinct from~$\sigma_1$ in the case of~$R'=R$; for the latter see \cref{fig:vrvc-case-c}. We denote by $C_j$, $j\in\{1,\ldots,a\}$, the facial cycle of~$\tau_j$. It is now important to notice that each cycle $C_j$ intersects at most two parts of $\ca R^0$, which follows from our ``multi-colored'' choice of $t_1,t_2,t_3$ and $R_1,R_2,R_3$ in all three cases. Furthermore, every two parts of $\ca R^0$ are together intersected by at most one of~$C_j$.
	
We next proceed similarly as in the separable case above. Let $G_j\subsetneq G_C$, $j\in\{1,\ldots,a\}$, be the strict subgraph of $G_C$ bounded by~$C_j$, and let $\ca R^0_j$ be the partition of $V(C_j)$ consisting of $V(C_j)\setminus V(C)$ and of the non-empty parts $X\cap V(C_j)$ over $X\in\ca R^0$. So, $|\ca R^0_j|\leq 3$. Therefore, by a recursive application of our algorithm, we may assume that each graph $G_j$ admits a good partition $\mathcal R_j \supseteq \ca R^0_j$ of $V(G_j)$, with $j=1,\ldots,a$.

We construct a partition $\ca R'\supseteq\ca R^0$ of $V(G_C)$ similarly as before; besides $\ca R^0$ we add the set $Z:=V(R^+)\setminus V(C)\not=\emptyset$ as whole, and the recursively obtained parts of each $G_j$ that do not touch $C_j$. Formally, $\ca R'=\ca R_0\cup \{Z\}\cup \bigcup_{j=1,\ldots,a}(\ca R_j\setminus\ca R^0_j)$. Note that $\ca R'$ is a partition of $V(G_C)$ ---in particular, each vertex of $G_C$ which is not on $C$ must belong either to $Z$ or to a part $X \in \ca R_j \setminus \ca R^0_j$, for some $j \in \{1,\dots,a\}$, by induction. We next show that the constructed partition $\ca R'$ is good.

\begin{claim}\label{claim:good-genral}
The partition $\ca R'$ constructed for the general case of $G_C$ is good.
\end{claim}
\begin{proof}
\cref{cl:inclusion} of \cref{lem:recurse} holds true for every $X\in\ca R'\setminus(\ca R_0\cup \{Z\})$ by recursion. For $X=Z$ we argue as follows. We choose $Z':=Z\setminus V(R_1\cup R_2\cup R_3)\subseteq R'$, and argue according to the Cases  \eqref{it:vrvc-case-a}--\eqref{it:vrvc-case-c} that we distinguished for $V(R)\cap V(C)$. In the Case~\eqref{it:vrvc-case-a}, i.e., $t_1,t_2\in V(C)$, we have $V(R_1\cup R_2)\cap Z=\emptyset$, and so $Z\setminus Z'=V(R_3)$ where $R_3$ is vertical in~$\sk(G)$ and $\lambda_{\ca L}(Z\setminus Z')=1$. In this case we also have   $|Z'|\leq|R|-3\leq h-3$, as desired. In the Cases~\eqref{it:vrvc-case-b} and \eqref{it:vrvc-case-c}, we similarly have that $Z\setminus Z'$ is made of $2$ and $3$ vertical paths, respectively, and the bounds on $\lambda_{\ca L}(Z\setminus Z')$ follow from that. In the Case~\eqref{it:vrvc-case-b} we also get $|Z'|\leq\floor{\frac12(|V(R)|-1)}-1\leq\floor{\frac12(h-1)}-1$ as desired, since the distance from $t_3$ to $V(C)$ on $R$ is not more than~$\frac12(|V(R)|-1)$. In the Case~\eqref{it:vrvc-case-c} we have that there are at most $\frac13|V(R)|$ vertices of color $3$ on $R$ (which stay in $Z'$ except the end of~$R_3$), and so $|Z'|\leq\floor{\frac13|V(R)|}-1\leq \floor{\frac h3}-1$.

We now turn the attention to the quotient graph $H:=H'$. Recall that we have recursively obtained the planar quotient graphs $H_j:=G_j/\ca R_j$, for $j \in \{1,\dots,a\}$, and we may assume, from the recursive invocation of \Cref{cl:II-triangle}, that each $H_j$ is a plane topological graph with the vertices stemming from the parts of $\ca R^0_j$ on the outer face. We again call these vertices stemming from $\ca R^0_j$ the {\em connectors} of~$H_j$, and we have from \Cref{cl:tripod} that no vertices of $H_j$ other than the connectors are adjacent to vertices of $H$ outside of $H_j$.

As previously, the graph $H$ consists of a $4$-clique $Q$ on the vertices which stem from the four parts of $\ca R^0\cup\{Z\}$, and of the union of the graphs $H_j$, with~$j=1,\ldots,a$, after identification of their connectors with the vertices of~$Q$. Let $z$ be the vertex of $Q$ that stems from the part $Z$ and $w_1,w_2,w_3$ be the vertices which stem from the parts $X_1,X_2,X_3$ of~$\ca R^0$. As already noted, none of the graphs $H_j$ contains all three $w_1,w_2,w_3$ (as its connectors), and for every pair from $\{w_1,w_2,w_3\}$, say $w_c$ and $w_d$, at most one of $H_j$, $j\in\{1,\ldots,a\}$, contains both $w_c,w_d$ among its connectors. In such case, $H_j$ is to be embedded in the triangular face $\{w_c,w_d,z\}$ of $Q$. Furthermore, if some $H_j$, $j\in\{1,\ldots,a\}$, contains only one of $w_1,w_2,w_3$ as its connector, say $w_b$, then $H_j$ can be embedded in any of the two triangular faces of $Q$ incident to the edge $\{w_b,z\}$.

Altogether, we have obtained a plane drawing of $H=G_C/\ca R'$ such that the vertices $w_1,w_2,w_3$ stemming from the parts of $\ca R^0$ are on the same triangular face. We have thus verified \Cref{cl:II-triangle}, and it remains to verify \Cref{cl:II-treewidth} in the aspect of treewidth of $H$. Again, we have recursively obtained a tree-decomposition $\ca T_j$ of $H_j$ for every $j \in \{1,\dots,a\}$, such that the clique of the vertices which stem from $\ca R_j^0$ is contained in a node $\nu_j$ of it. We create a new decomposition $\ca T$ of $H$ from the disjoint union of $\ca T_j$ over $j=1,\ldots,a$ by adding a new node $\nu$, holding the bag of vertices $V(Q)$ and adjacent exactly to all $\nu_j$. Further, for $i=1,\dots,k$, we rename each vertex of $\mathcal{T}_j$ that stems from $X^j_i$, with $i=1,\dots,a$, as the vertex in $Q$ that stems from $X_i$ and, for $i=1,\dots,a$, we rename each vertex of $\mathcal{T}_j$ that stems from $Y_i$ as the vertex in $Q$ that stems from $Y$. This is a valid tree-decomposition by  \Cref{cl:tripod}, and it is of width $3$ by \Cref{cl:II-treewidth} and the fact that $|V(Q)|-1=3$.
\end{proof}

We conclude the proof of \cref{lem:recurse} by discussing the time complexity of our algorithm. All the necessary work needed at each recursive call can be easily implemented to run in $O(|V(G_C)|)$ time, by adopting the approach of~\cite{DBLP:journals/corr/abs-2004-02530} to compute $R_1$, $R_2$ and $R_3$ in the general case (provided that $G$ is a topological $h$-framed graph). Since the total number of recursive calls is at most linear in $|V(G_C)|$, the total running time is thus quadratic in $|V(G_C)|$.
\end{proof}

\begin{proof}[Proof of~\cref{thm:prodstruct-h}.]
Let $C$ denote the cycle bounding the outer face of $\sk(G)$, which, by a possible homeomorphism of the sphere, may be assumed to satisfy $|V(C)|\geq3$. Based on the BFS tree $T$ of $\sk(G)$ rooted in a vertex~$r\in V(C)$, we define the following partition~$\mathcal R^0$ of $C$: We split $C$ into a path $P_1$ only consisting of the vertex $r$, and two paths $P_2$ and $P_3$ of lengths at most $\floor{\frac{h-1}2}$ and $\floor{\frac{h}2}$, respectively. This defines the initial setup of our algorithm which allows us to invoke \cref{lem:recurse} (e.g., with $q=2$ and using trivial single-vertex vertical paths in $P_2$ and $P_3$). Then, we set $\mathcal R^0 = \{V(P_1),V(P_2),V(P_3)\}$ and apply the algorithm given in the proof of \cref{lem:recurse}. This way we obtain a good partition $\ca R'$ of $V(G_C)=V(G)$ and graph $H' := G_C /{\ca R'}$ in $O(|V(G)|^2)$ time.
Note that, in general, $G_C\not=G$ as $G$ may have edges drawn in the outer face (bounded by~$C$) of $\sk(G)$. However, by setting $H = H'$, we guarantee all edges of $G$ in the outerface of~$\sk(G)$ are ``captured'', since the quotient graph $H'$ anyway contains a triangle on the vertices that stem from~$\ca R_0$. In fact, we have just obtained the graph $H$ with the desired properties, i.e., $H$ is planar and of $\tw(H) \leq 3$.
	
What remains to prove is that $G$ indeed is a subgraph of the strong product $H \boxtimes P \boxtimes K_{\Ksize}$ for some path $P$. Recall that the number of layers of the layering $\ca W$ is $\ell+1$, and that $\ca W$ was obtained by merging consecutive $\lfloor\frac{h}2\rfloor$-tuples of layers of~$\ca L$. We set $P$ to be the path on $\ell+1$ vertices denoted in order by $p_0,p_1,\ldots,p_\ell$. To a vertex $v\in V(G)$, we assign the pair $(t,p_i)$ where $t\in V(H)$ if $t$ stems from the part of $\ca R'$ that $v$ belongs to, and $v\in W_i\in\ca W$. This assignment is sound and unique. If $vv'\in E(G)$ is any edge of $G$, and $v$ and $v'$ are assigned the pairs $(t,p_i)$ and $(t',p_j)$ as above, then $tt'\in E(H)$ or $t=t'$ since $H=G/\ca R'$ is the quotient graph, and $p_ip_j\in E(P)$ or $i=j$ since $\ca W$ is a layering of $G$. Using \cref{cl:inclusion} of \Cref{lem:recurse}, we furthermore estimate, for every part $X \in \mathcal R'$ and its $X'\subseteq X$ (cf.~\cref{cl:inclusion}),
\begin{align*}
	\lambda_{\ca W}(X) &\leq |X'| + \lambda_{\ca L}(X\setminus X')\cdot \floor{h/2}
\\	&\leq \max\big( h-3+\floor{h/2},\> \floor{h/2}-1+2\floor{h/2},\> \floor{h/3}-1+3\floor{h/2} \big)
\\	&\leq 3\floor{h/2}+\floor{h/3}-1 \,,
\end{align*}
and hence at most $\Ksize$ vertices of $G$ are assigned to the same pair $(t,p_i)$.
This concludes the proof that $\simple{G}\subseteq H \boxtimes P \boxtimes K_{\Ksize}$.
\end{proof}

\noindent We next present a variant of \cref{thm:prodstruct-h}, which reduces the size of the clique in the product by replacing the path with a power of it.

\begin{theorem}\label{thm:prodstruct-magic}
Let $G$ be an $h$-framed graph (where $G$ is not necessarily simple). Then $\simple{G}$ is a subgraph of the strong product of three graphs $H \boxtimes P^{\floor{h/2}} \boxtimes K_{max(3,h-2)}$, where $H$ is a planar graph with $\tw(H) \leq 3$ and $P$ is a path. 
\end{theorem}

\begin{proof}
Recall that $\ca L = (V_0,V_1,\ldots,V_{b})$ is a BFS layering of the skeleton $\sk(G)$, and thus every edge of $G$ has ends in parts $V_i,V_j\in\ca L$ such that~$|i-j|\leq \floor{h/2}$. Hence, we may choose $P$ as the path on $b+1$ vertices $(p_0,p_1,\ldots,p_b)$, use $P^{\floor{h/2}}$, and assign each vertex $v\in V(G)$ to the pair $(t,p_i)$ where $t\in V(H)$ if $t$ stems from the part of $\ca R'$ that $v$ belongs to, and $v\in V_i\in\ca L$. Now, the number of vertices of $G$ assigned to the same pair $(t,p_i)$ (where $t$ stems from a part $X$) is at most
$\lambda_{\ca L}(X)\leq |X'| + \lambda_{\ca L}(X\setminus X') \leq \max( h-3+1,\> \floor{h/2}-1+2,\> \floor{h/3}-1+3 \big)= max(3, h - 2)$. This concludes that $\simple{G}\subseteq H \boxtimes P^{\floor{h/2}} \boxtimes K_{\max(3,h-2)}$.
\end{proof}

\section{Consequences of the Product Structure}
\label{se:consequences}

As mentioned in the introduction, Dujmovic et al.~\cite{DBLP:journals/corr/abs-1907-05168} have derived upper bounds on the queue number, on the non-repetitive chromatic number, and on the $p$-centered chromatic number of $k$-planar and $k$-map graphs exploiting \cref{th:k-map-planar}. In the following, we present our improvements to each of these problems.

\myparagraph{Queue number.} 
A \emph{queue layout} of a graph $G$ is a linear order $\sigma$ of the vertices of $G$ together with an assignment of its edges to sets, called \emph{queues}, such that no two edges in the same set nest. 
The \emph{queue number} $\qn(G)$ of a graph $G$ is the minimum number of queues over all queue layouts of $G$.
In~\cite{DBLP:journals/jacm/DujmovicJMMUW20}, Dujmović et al.\ have proved the following useful lemma concerning the queue number of graphs that can be expressed as subgraphs of the strong product of a path $P$, a graph $H$ with queue number $\qn(H)$, and a clique $K_\ell$ on $\ell$ vertices.

\begin{lemma}[Dujmovic et al.~\cite{DBLP:journals/jacm/DujmovicJMMUW20}]\label{lem:threegraphs}
If $G\subseteq P \boxtimes H \boxtimes K_\ell$ then 
$\qn(G) \leq  3 \ell \, \qn(H) + \floor{\frac{3}{2}\ell}$.
\end{lemma}

Combining \cref{lem:threegraphs} and \cref{th:k-map-planar}\ref{th:k}, together with the fact that the queue number of planar $3$-trees is at most $5$~\cite{DBLP:journals/algorithmica/AlamBGKP20}, Dujmović, Morin, and Wood showed the first constant upper bound on the queue number of $k$-planar graphs~\cite{DBLP:journals/corr/abs-1907-05168}, thus resolving a long-standing open question. Analogously, by combining \cref{lem:threegraphs,thm:prodstruct-h}, we obtain the following.

\begin{theorem}\label{th:queue-h-framed}
The queue number of $h$-framed graphs is at most $$15 \cdot (\Ksize) + \floor{\frac{3}{2}(\Ksize)}.$$
\end{theorem}

Dujmovic et al.~\cite{DBLP:journals/jacm/DujmovicJMMUW20} first showed the queue number of $k$-map graphs is at most $2(98(k+1))^3$. Later, by combining \cref{th:k-map-planar}\ref{th:map} and \cref{lem:threegraphs}, Dujmovic et al.~\cite{DBLP:journals/corr/abs-1907-05168} improved this bound to $32225k(k-3)$. By \cref{th:queue-h-framed}, we improve these bounds to $61k$ by 
exploiting the fact that these graphs are subgraphs of $2k$-framed graphs~\cite{DBLP:journals/corr/abs-2003-07655,DBLP:conf/compgeom/BekosLGGMR20}. 

\begin{corollary}\label{thm:queuenumber-kmap}
The queue number of $k$-map graphs is at most $$15 \cdot (3k + \floor{2k/3} -1) + \floor{\frac{3}{2}(3k + \floor{2k/3} -1)}.$$
\end{corollary}

For $h \in {4,5}$, \cref{th:queue-h-framed} gives us an upper bound of $95$.  Since any $1$-planar graph can be augmented to a (not-necessarily simple) $4$-framed graph~\cite{DBLP:conf/gd/AlamBK13}, \cref{th:queue-h-framed} improves the current best upper bound of $1$-planar graphs from $495$ to $95$. Since any optimal $2$-planar graph is $5$-framed, \cref{th:queue-h-framed} improves the current best upper bound on their queue number from $3267$  (derived from \cref{lem:threegraphs} with $\ell = 198$) to $95$. Next, we show a generalization of \cref{lem:threegraphs} that allows further improvements.

\begin{lemma}lemma\label{lem:queue-power}
If $G\subseteq H \boxtimes P^{i} \boxtimes K_\ell$ then $\qn(G) \leq i\ell+(2i+1)\ell\qn(H) +  \floor{\frac{\ell}{2}}$.
\end{lemma}
\begin{proof}
For convenience, let ${\cal X}_j=H \boxtimes P^j \boxtimes K_\ell$, with $1 \leq j \leq i$. Observe that the graphs ${\cal X}_j$, with $1 \leq j \leq i$,  have the same vertex set and ${\cal X}_j \subseteq {\cal X}_{j+1}$, with $1 \leq j < i$. Let $P=(p_1,p_2,\dots,p_z)$ and let $\langle x_1,x_2,\dots,x_q \rangle$ be the vertex ordering of a $\qn(H)$-queue layout of $H$. We set $V_{a,b}:=\{v_{a,b}\} \times V(K_\ell)$, where $v_{a,b}$ denotes the vertex of $V(H) \times V(P)$ that stems from the vertex $x_a$ of $H$ and the vertex $p_b$ of $P$. Note that, the sets $V_{a,b}$ form a partition of $V({\cal X}_1)$. The following property follows from the proof of \cref{lem:threegraphs} given in~\cite{DBLP:journals/jacm/DujmovicJMMUW20}.

\begin{property}[Vertex order of \cref{lem:threegraphs}]\label{prop:invariant-order}
The queue layout of ${\cal X}_1$ in the proof of \cref{lem:threegraphs} is such that, for any two vertices $u \in V_{a,b}$ and $v \in V_{c,d}$, it holds that $u$ precedes $v$ in such a layout if and only if one of the following holds: Either $b < d$ or $c=d$ and $a < c$.
\end{property}

Our proof is by induction on $i$. In particular, we will show that ${\cal X}_i$ has a queue layout whose vertex order $\sigma$ satisfies \cref{prop:invariant-order} and uses at most $i\ell+(2i+1)\ell\qn(H) +  \floor{\frac{\ell}{2}}$ queues. In the base case $i=1$ and the result follows directly from \cref{lem:threegraphs}. 
Assume now that $i>1$. Let $\Delta_i$ be the graph obtained by removing from $P^i$ the edges that it shares with $P^{i-1}$, i.e., $\Delta_i = (V(P),E(P^i) \setminus E(P^{i-1}))$. Clearly, two vertices are adjacent in $\Delta_i$ if and only if they are at distance $i$ in $P$.
Observe now that ${\cal X}_i$ is the union of ${\cal X}_{i-1}$ and $H \boxtimes \Delta_{i} \boxtimes K_\ell$. By induction, we have that ${\cal X}_{i-1}$ admits a queue layout $\Gamma$ whose vertex order satisfies \cref{prop:invariant-order} and uses at most $(i-1)\ell + (2i-1)\ell \qn(H) + \floor{\frac{\ell}{2}}$ queues. Therefore, in order to prove the statement, it suffices to show that the edges of $H \boxtimes \Delta_{i} \boxtimes K_\ell$ can be added in $\Gamma$ by using at most $\ell + 2 \qn(H)\ell$ queues. To this aim, we classify the edges of this graph into three sets $E_|$, $E_{\textbackslash}$, and $E_/$. Namely, for each edge $(u,v)$ with $u \in V_{a,b}$ and $v \in V_{c,d}$, 
we have that:
\begin{itemize}
\item $(u,v) \in E_|$, if $b=d$;
\item $(u,v) \in E_{\textbackslash}$, if $b<d$ and $a<c$; and
\item $(u,v) \in E_/$, if $b<d$ and $a>c$.
\end{itemize}

First, we show that the edges in $E_|$ can be assigned to at most $\ell$ queues. For this, we recall that the number of queues in a queue layout coincides with the size of its largest rainbow~\cite{DBLP:journals/siamcomp/HeathR92}, where a \emph{rainbow} is a set of pairwise-nesting independent edges in a linear order of the vertices.
Namely, let $(u,v)$ and $(u',v')$ be  two edges in $E_|$ with $u \in V_{a_1,b}$, $v \in V_{a_2,b}$, $u' \in V_{a_1',b'}$ and $v' \in V_{a_2',b'}$. Assuming w.l.o.g. that $a_1 \leq a_1'$ holds, it follows that these two edges nest in $\sigma$, only if 
$b=b'$, $a_1 = a_1'$, and $a_2 = a_2'$.
Since each of the sets $V_{a_1,b}$ and $V_{a_2,b}$ contains at most $\ell$ vertices and since the vertices in $V_{a_1,b}$ precede the vertices of $V_{a_2,b}$ in $\sigma$, we have that the maximum rainbow formed by edges in $E_|$ has size at most $\ell$ in $\sigma$. Thus, $\ell$ queues suffice to embed all such edges in $\Gamma$~\cite{DBLP:journals/siamcomp/HeathR92}.
Second, we that all the edges in $E_{\textbackslash}$ can be assigned to at most $\ell \qn(H)$ queues. Since the proof that the edges in $E_/$ can be assigned to at most $\ell \qn(H)$ queues is analogous, this concludes the proof of the lemma.
Consider a partition of $E_{\textbackslash}$ into sets $E_i$, with $1 \leq i \leq \qn(H)$, such that $E_i$ contains all the edges $(u,v)$ of $E_{\textbackslash}$ such $u \in V_{a, c}$, $v \in V_{c,d}$, and $(a,c)$ belong to the $i$-th queue of $H$.
Next, we show that the edges in each set $E_i$ can be assigned to at most $\ell$ new queues in $\Gamma$. 
Consider any two nesting edges $(u,v)$ and $(u',v')$ in $E_i$ with $u \in V_{a_1,b_1}$, $v \in V_{a_2,b_2}$, $u' \in V_{a_1',b_1'}$ and $v' \in V_{a_2',b_2'}$.
By the definition of $E_{\textbackslash}$, we have that $a_1 < a_2$, $a_1' < a_2'$, $b_1 < b_2$, and $b_1' < b_2'$. 
Since $(u,v),(u',v') \in E(\Delta_i)$ they have the same span and thus it follows that 
$(u,v)$ and $(u',v')$ nest in $\sigma$ only if
(i) $b_1=b_1'$ and $b_2=b_2'$ (that is, $u$ and $u'$ (resp.\ $v$ and $v'$) stem from the same vertex of $\Delta_i$) and (ii)
$a_1=a_1'$ or $a_2=a_2'$ (that is, at least one of the pairs $u,u'$ and $v,v'$ stem from the same vertex of $H$)~\cite{DBLP:journals/jacm/DujmovicJMMUW20}. 
Since each of the sets $V_{a_1,b_1}$ and $V_{a_2,b_2}$ contains at most $\ell$ vertices and since the vertices in $V_{a_1,b_1}$ precede the vertices of $V_{a_2,b_2}$ in $\sigma$, we have that the maximum rainbow formed by edges in $E_i$ has size at most $\ell$ in $\sigma$. Thus, $\ell$ queues suffice to embed all such edges in $\Gamma$~\cite{DBLP:journals/siamcomp/HeathR92}, and the proof is concluded.
\end{proof}

\noindent \cref{lem:queue-power} in conjunction with \cref{thm:prodstruct-magic} yields a quadratic (in $h$) upper bound on the queue number of $h$-framed graphs. However, for $h \leq 5$, it implies an improved bound on the queue number of $1$-planar and optimal $2$-planar graph, which we summarize in the following.

\begin{theorem}\label{thm:queuenumber-1p-o2p}
The queue number of $1$-planar and optimal $2$-planar graphs is at most $81$.
\end{theorem}

\myparagraph{Non-repetitive chromatic number.} 
An \emph{$r$-coloring} of a graph $G$ is a function $\phi: $V(G)$ \rightarrow [r]$. A path $p=(v_1,v_2,\dots,v_{2\tau})$ is \emph{repetitively colored} by $\phi$ if $\phi(v_i)= \phi(v_{i + \tau})$ for $i=1,2,\dots,\tau$. A coloring $\phi$ of $G$ is \emph{non-repetitive} if no path of $G$ is repetitively colored by $\phi$. Clearly, non-repetitive colorings are \emph{proper}, i.e., $\phi(u) \neq \phi(v)$ if $u$ and $v$ are adjacent in $G$. The \emph{non-repetitive chromatic number} $\pi(G)$ of $G$ is the minimum integer $r$ such that $G$ admits a non-repetitive $r$-coloring. In~\cite{2020-non-repetitive}, Dujmovic et al.\ developed the following lemma to upper-bound the non-repetitive chromatic number of graphs that can be expressed as subgraphs of the strong product of a path $P$, a graph $H$ with $\tw(H)$, and a clique $K_\ell$ on $\ell$~vertices.

\begin{lemma}[Dujmovic et al.~\cite{2020-non-repetitive}]\label{lem:repetitive}
If $G\subseteq P \boxtimes H \boxtimes K_\ell$, then $\pi(G) \leq 4^{tw(H)+1} \cdot \ell$.
\end{lemma}

Using \cref{lem:repetitive} and \cref{th:k-map-planar}\ref{th:one}, Dujmovic et al.~\cite{DBLP:journals/corr/abs-1907-05168} provide an upper bound of $7680$ on the non-repetitive chromatic number of $1$-planar graphs. 
Since $1$-planar graphs are $4$-framed, we improve this bound to $1536$. 
Using the fact that $h$-framed graphs are $O(h^2)$-planar and the bound derived by combining \cref{lem:repetitive} and \cref{th:k-map-planar}\ref{th:map}, we may conclude that their non-repetitive chromatic number is bounded (but the bound is super-polynomial in $h$). By \cref{lem:repetitive} and \cref{thm:prodstruct-h}, we obtain an upper bound for $h$-framed graphs that is linear in $h$. Also, since $k$-map graphs are subgraphs of $2k$-framed graphs~\cite{DBLP:journals/corr/abs-2003-07655,DBLP:conf/compgeom/BekosLGGMR20}, their non-repetitive chromatic number is linear in $k$; this is an improvement over the quadratic bound given in~\cite{DBLP:journals/corr/abs-1907-05168}. Specifically, we get the~following.

\begin{corollary}\label{thm:non-repetitive}
For a graph $G$, it holds:
\begin{enumerate}[label={(\roman*)}, ref=\roman*]
\item $\pi(G) \leq 4^4 \cdot 6$, if $G$ is $1$-planar,
\item $\pi(G) \leq 4^4 \cdot (3k + \floor{2k/3} - 1) $, if $G$ is $k$-map, and
\item $\pi(G) \leq 4^4 \cdot(\Ksize)$, if $G$ is $h$-framed.
\end{enumerate}
\end{corollary}

\myparagraph{$p$-centered chromatic number.} 
For any $c,p \in \mathbb{N}$ with $c \geq p$, a $c$-coloring of a graph $G$ is \emph{$p$-centered} if, for every connected component $X$ of $G$, at least one of the following holds: (i) the vertices of $X$ are colored with more than $p$ colors, or (ii) there exists a vertex $v$ of $X$ that is assigned a color different from the ones of the remaining vertices of $X$. For any graph $G$, the \emph{$p$-centered chromatic number} $\chi_p(G)$ of $G$ is the minimum integer $c$ such that $G$ admits a $p$-centered $c$-coloring. 
Pilipczuk and Siebertz showed that $\chi_p(H) \leq \binom{p-t}{t}$ for any graph $H$ of treewidth $t$~\cite{DBLP:journals/jctb/PilipczukS21}. However, for the special case in which $H$ is planar and has treewidth at most $3$, Debski et al.~\cite{DBLP:conf/soda/DebskiFMS20} showed that $\chi_p(H) \in O(p^2 \log p)$. Also, Debski et al.~\cite{DBLP:conf/soda/DebskiFMS20} and Dujmovic et al.~\cite{DBLP:journals/corr/abs-1907-05168} showed that for any subgraph $G$ of $H \boxtimes P \boxtimes K_\ell$ it holds $\chi_p (G) \leq \ell(p+1)\chi_p(H)$. For convenience, hereafter we combine these results to upper-bound the $p$-centered chromatic number of graphs that can be expressed as subgraphs of the strong product of a path $P$, a graph $H$ with treewidth $\tw(H)$, and a clique $K_\ell$ on $\ell$ vertices.

\begin{lemma}[\cite{DBLP:conf/soda/DebskiFMS20,DBLP:journals/corr/abs-1907-05168,DBLP:journals/jctb/PilipczukS21}]\label{lem:centered}
If $G\subseteq P \boxtimes H \boxtimes K_\ell$, where $\tw(H) \leq 3$, it holds:
\begin{enumerate}[label={(\roman*)}, ref=\roman*]
\item \label{lem:centered-planar} 
$\chi_p(G) \in O(\ell p^3 \log p)$, if $H$ is planar, and
\item \label{lem:centered-non-planar} 
$\chi_p(G) \in O(\ell p^4)$, if $H$ is not planar.
\end{enumerate}
\end{lemma}

\noindent By \cref{lem:centered}.\ref{lem:centered-non-planar} and \cref{th:k-map-planar}, Dujmovic et al.~\cite{DBLP:journals/corr/abs-1907-05168} showed that $\chi_p(G) \in O(p^4)$ if $G$ is a $1$-planar graph, $\chi_p(G) \in O(k^2 p (p+k^3)^{k^3})$ if $G$ is a $k$-planar graph, and $\chi_p(G) \in O(k^2 p^{10})$ if $G$ is a $k$-map graph.
Since $h$-framed graphs are $O(h^2)$-planar, their $p$-centered chromatic number is $O(h^4 p (p+h^6)^{h^6})$. 
By exploiting \cref{thm:prodstruct-h}  
and \cref{lem:centered}.\ref{lem:centered-planar}, we get the next. 

\begin{corollary}\label{thm:centered}
For a graph $G$, it holds:
\begin{enumerate}[label={(\roman*)}, ref=\roman*]
\item $\chi_p(G) \in O(p^3 \log p)$, if $G$ is~$1$-planar,
\item $\chi_p(G) \in O(k  p^3 \log p)$, if $G$ is $k$-map, and
\item $\chi_p(G) \in O(h p^3 \log p)$, if $G$ is $h$-framed.
\end{enumerate}
\end{corollary}

\section{Bounding Twin-width}\label{se:twinwidth}

Besides the direct consequences of the product structure theorem(s) surveyed in \cref{se:consequences}, the construction presented in \cref{sec:getstructure} has another strong implication described next.

Consider only simple graphs for the coming definition.\footnote{In general, the concept of twin-width is defined for binary relational structures of a finite signature, and so one may either define the twin-width of a multigraph as the twin-width of its simplification, or allow only bounded multiplicities of edges and use the more general matrix definition of twin-width.} A \emph{trigraph} is a simple graph $G$ in which some edges are marked as {\em red}, and we then naturally speak about \emph{red neighbors} and \emph{red degree} in~$G$. We denote the set of red neighbors of a vertex $v$ by $N_r(v)$. For a pair of (possibly not adjacent) vertices $x_1,x_2\in V(G)$, we define a \emph{contraction} of the pair $x_1,x_2$ as the operation creating a trigraph $G'$ which is the same as $G$ except that $x_1,x_2$ are replaced with a new vertex $x_0$ whose full neighborhood is the union of neighborhoods of $x_1$ and $x_2$ in $G$, that is, $N(x_0)=(N(x_1)\cup N(x_2))\setminus\{x_1,x_2\}$, and the red neighbors $N_r(x_0)$ of $x_0$ inherit all red neighbors of $x_1$ and of $x_2$ and those in $N(x_1)\oplus N(x_2)$, that is, $N_r(x_0)=\big((N_r(x_1)\cup N_r(x_2))\setminus\{x_1,x_2\}\big)\cup\big(N(x_1)\oplus N(x_2)\big)$, where $\oplus$ denotes the~symmetric~difference.

A \emph{contraction sequence} of a trigraph $G$ is a sequence of successive contractions turning~$G$  into a single vertex, and its \emph{width} is the maximum red degree of any vertex in any trigraph of the sequence. The \emph{twin-width} is the minimum width over all possible contraction sequences (where for an ordinary graph, we start with the same trigraph having no red edges). As noted already in the pioneering paper on this concept~\cite{DBLP:journals/jacm/BonnetKTW22}, the twin-width of $k$-planar graphs is bounded for any fixed $k$ by means of FO transductions (which, therefore, gives a not-even-asymptotically expressible bound). Explicit asymptotic bounds for the twin-width of $k$-planar graphs (albeit with $O(k)$ in the exponent, and so not giving an explicit number for e.g.~$k=1$) are in~\cite{DBLP:journals/corr/abs-2202-11858} (with a generalization to higher surfaces), and specially for planar graphs, the current upper bounds on twin-width are $583$ in~\cite{DBLP:journals/corr/abs-2202-11858} and improved $183$ in~\cite{DBLP:journals/corr/abs-2201-09749}. It is worth to mention that both~\cite{DBLP:journals/corr/abs-2202-11858,DBLP:journals/corr/abs-2201-09749}, more or less explicitly, use the product structure machinery of planar graphs. We give an improved bound on the twin-width of planar graphs, and new explicit (non-asymptotic) bounds on the twin-width of $h$-framed and $1$-planar graphs.

\begin{theorem}\label{cref:ttwin-width-planar}
The twin-width of a simple planar graph $G$ is at most~$37$.
\end{theorem}
\begin{proof}
We use the refined planar product structure machinery from~\cite{DBLP:journals/corr/abs-2108-00198}. In fact, the outer skeleton of our proof is the same as in~\cite[Section~4]{DBLP:journals/corr/abs-2201-09749}, but we use a different invariant ``inside''. As in \cref{sec:getstructure}, we start with a technical setup (now simplified).

Let $G$ be a simple planar graph, and let $G^+\supseteq G$ be an arbitrary plane triangulation on the same vertex set~$V(G^+)=V(G)$. Let $S$ be the cycle of the outer face of $G^+$ (so, $S$ is a triangle) and $t\in V(S)$. Consider a BFS tree $T$ of $G^+$ (not of~$G$) rooted in $t$ and the {BFS layering} $\ca L=(V_0,V_1,\ldots)$ of $G^+$ such that $V_i$ contains all vertices of $G^+$ at distance $i$ from $t$. A \emph{partial contraction sequence} of $G$ is defined in the same way as a contraction sequence of $G$, except that it does not have to end with a single-vertex graph. Such a sequence is \emph{$\ca L$-respecting} if every step contracts only pairs belonging to $V_j$ for some~$j$. To be formally precise, when $G'$ is a trigraph resulting from an $\ca L$-respecting partial contraction sequence of $G$, we should write $V_j[G']$ for the set that stems from $V_j=V_j[G]$. The following is immediate and useful to notice:

\begin{claim}\label{cl:twwlayers}
Let $G'$ be a trigraph obtained from $G$ by an $\ca L$-respecting partial contraction sequence.
If there is a red edge in $G'$ from $V_j[G']$ to $V_{j'}[G']$, then~$|j-j'|\leq1$.
\end{claim}

\noindent The \emph{$G$-neighborhood of a vertex $x$ in a set $Y$} is the subset $Y'\subseteq Y$ of those vertices adjacent to~$x$ in~$G$. \cref{cref:ttwin-width-planar} will follow if we (recursively) prove the following:

\begin{claim}\label{cl:twinrecur}
Let $C$ be a cycle in $G^+$, and let $G_C$ be the subgraph of $G$ formed by the vertices and edges along and in the interior of~$C$ (i.e., $G_C$ is the subgraph of $G$ bounded by~$C$), such that $U:=V(G_C)\setminus V(C)$ does not contain~$t$. Let $a,b\in V(C)$ be arbitrary. Assume that, for some $1\leq k\leq 5$, the vertex set $V(C)$ is partitioned into $k$ pairwise disjoint vertical paths $P_1,\ldots,P_k\subseteq C$. Then there exists an $\ca L$-respecting partial contraction sequence of $G$ which contracts only vertex pairs in $U$, results in a trigraph $G_0$, and satisfies the following:
\vspace*{-1ex}%
\begin{enumerate}[label={\alph*)}, ref=(\alph*)]
\item Each vertex of $V(C)$ has red degree at most $12$ along the sequence, and no red edge is incident to either of~$a,b$;
\item each vertex that stems from $U$ along the sequence has red degree at most $37$; and
\item no two vertices of $V_{j}^0:=(V(G_0)\setminus V(C))\cap V_j[G_0]$, for $j=1,2,\ldots$, have the same $G_0$-neighborhood in~$\{a,b\}$.
So, in particular, $|V_j^0|\leq4$.\label{it:twinrecurc}
\end{enumerate}
\end{claim}

Indeed, we can apply \cref{cl:twinrecur} to the outer triangle $S$, obtaining the trigraph $G_0$ along an $\ca L$-respecting sequence, and then successively (in an arbitrary choice and order of vertex pairs) contract $V_0=\{t\}$ with $V_1[G_0]$ into one vertex, then with $V_2[G_0]$ into one vertex, with $V_3[G_0]$, and so on. By \cref{cl:twwlayers}, we have that only three layers of $\ca L$ contribute to the maximum red degree, which is thus at most $(4+2)+4+4-1=13<37$ ($+2$ stands for the two vertices of $V_1$ in~$S$) in this remaining contraction sequence of~$G_0$.

\smallskip
The rest of the proof is devoted to proving \cref{cl:twinrecur} by induction on~$|U|$.
If $U=\emptyset$, then we are done with the empty contraction sequence and~$G_0=G$. Otherwise, we start with a few more basic observations. Since we are contracting only vertex pairs in $U$, every red edge of $G_0$ must have one or both ends in $U$ (precisely, in a vertex that stems from $U$, but we slightly abuse the notation to keep it simple). Furthermore, if a red edge of $G_0$ starts in a vertex $w\not\in U$, then $w$ must have been adjacent to some vertex of $U$ in~$G$. Together with the assumed planarity of $G^+\supseteq G$ and its cycle $C$ bounding $U$ in the interior, this implies that there is no red edge in $G_0$ incident to a vertex of $V(G)\setminus V(G_C)$, i.e., no red edge ``crosses'' the cycle $C$ in~$G_0$. Although, note that $G_0$ is not necessarily planar (in the interior of~$C$).

We use the decomposition argument of~\cite{DBLP:journals/corr/abs-2108-00198} (exactly as~\cite[Section~4]{DBLP:journals/corr/abs-2201-09749} does). In this argument, we get a triangle $R_0$ of $G^+$, and three pairwise disjoint vertical paths $R_1,R_2,R_3\subseteq G^+$ (possibly degenerate) from the three vertices of $R_0$ to $V(C)$, such that the plane graph $C\cup R$, where $R:=R_0\cup R_1\cup R_2\cup R_3$, has at most three inner faces (other than $R_0$) bounded by the cycles $D_1,D_2,D_3\subseteq C\cup R$, and each of $D_1,D_2,D_3$ intersects at most three of the paths $P_1,\ldots,P_k$ (and two of $R_1,R_2,R_3$). Let $G_{D_i}$, $i\in\{1,2,3\}$, be the subgraph of $G$ bounded by~$D_i$. For $i=1,2,3$, let $r_i$ be the end of $R_i$ on~$C$, and, up to symmetry, assume that $R_i$ is the path disjoint from (``opposite to'') $D_i$.

We are going to apply \cref{cl:twinrecur} independently to each of $D_1,D_2,D_3$ with a suitable choice of the two vertices~$a,b$ as follows. If the given $a,b$ belong to the same cycle, say $a,b\in D_1$ up to symmetry, then we choose $a_1:=a$, $b_1:=b$, $a_2:=r_1$, $b_2:=r_3$ and $a_3:=r_1$, $b_3:=r_2$. Otherwise, if, up to symmetry, $a\in D_2$ and $b\in D_3$, then we choose $a_2:=a$, $b_2:=r_1$, $a_3:=r_1$, $b_3:=b$ and $a_1:=r_2$, $b_1:=r_3$. So, for $i=1,2,3$, we apply \cref{cl:twinrecur} to $D_i$ in place of $C$ and to $a_i,b_i$ as previously. This way we get partial contraction sequences of $G$ down to~$G_i$. We denote (cf.~\cref{it:twinrecurc}) by $V_{j}^i:=(V(G_i)\setminus V(D_i))\cap V_j[G_i]$ the contracted layers in the interior of~$D_i$. We compose these three partial sequences one after another (as there is no conflict or dependence between them), giving us a partial contraction sequence of $G$ to a trigraph $G_4$.

Observe that along the composed sequence from $G$ to $G_4$, only vertices of $V(R)$ could have received red edges from distinct applications of \cref{cl:twinrecur} in the previous paragraph, and these possible duplicities will be managed in our coming arguments. Specially, every vertex of $V(C)$ has again red degree at most $12$ -- this also holds for the vertices $r_1,r_2,r_3$ thanks to our right choice of $a_i,b_i$ (which essentially ``prevents'' $r_1,r_2,r_3$ from getting more red edges from concurrent recursive calls). Every vertex of $V(R)$ has red degree at most $2\cdot12=24<37$ along the sequence, and every other vertex in the interior of $C$ at most~$37$, which is direct from \cref{cl:twinrecur}. One can say even more; since every vertex of $V(D_i)$ is adjacent to at most three of the layers $V_{j-1}^i,V_{j}^i,V_{j+1}^i$ for some $j$ (it is a BFS layering of $G^+\supseteq G$) and $|V_{j-1}^i|+|V_{j}^i|+|V_{j+1}^i|\leq3\cdot4=12$, we may as well assume (the ``worst case'' scenario) that all edges from $V(D_i)$ to the rest of $G_{D_i}$ are red, except for the edges of $a_i$ and~$b_i$. Furthermore, recall that the (up to eight) paths in $\ca P:=\{P_1,\ldots,P_k,\,R_1,R_2,R_3\}$ of $G^+$ are vertical, and hence each of them intersects every BFS layer $V_j[G_4]$ in at most one vertex.

Subsequently, we will continue with a partial contraction sequence from $G_4$ to the desired outcome $G_0$, in order to finish our claim. This sequence will proceed in two {\em stages}, the first one contracting $V(R_1)$ and the interior vertices of $D_2$ and $D_3$ together (taking advantage of the fact the $r_1$ has no red edge in $G_4$), and then the second stage contracting the rest with $V(R_2)\cup V(R_3)$ and the interior vertices of $D_1$. Precisely, if $a,b\in D_1$ (cf.~the above case distinction), we do:%
\vspace*{-1ex}%
\begin{enumerate}
\item In the first stage, we proceed layer-by-layer with $j=1,2,3,\ldots$. We first contract possible pairs of vertices in $V_j^2$ with the same adjacency relation to $b_2=r_3$; this adds up to $2$ red edges incident to $a_2=r_1$ and none to $b_2$ (and recall that all other edges to $V(D_i)$ are already assumed red). Likewise, we contract possible pairs of vertices in $V_j^3$ with the same adjacency relation to $b_3=r_2$, again adding at most $2$ red edges to $b_2=r_1$.

Then we have at most $2$ vertices in each of $V_j^2$ and $V_j^3$ left, plus one vertex in $V_j\cap V(R_1)$, and we contract the remaining pairs among them of the same $G_4$-neighborhood to~$\{r_2,r_3\}$.

Let $D'$ denote the sub-cycle of $C\cup R_2\cup R_3$ bounding $G_{D_2}\cup G_{D_3}$. The red degrees on $D'$ do not increase in the described contractions, except that $r_1$ receives up to $4$ red edges from each of three layers possibly adjacent to $r_1$, and this gives the red degree of at most $0+3\cdot4=12$ for~$r_1$, which is good. It is more delicate to bound the red degrees of the interior vertices of $D'$ along our partial sequence (in this first stage):
\begin{itemize}
\vspace*{-1ex}%
\item When processing layer $j$, there are at most $4$ vertices in layer $j-1$ in the interior of~$D'$, and at most $7$ in layer $j-1$ on $D'$ (one from every vertical path in $\ca P\setminus\{R_1\}$).
\item In the current layer $j$, before the first contraction there, every interior vertex has at most $4-1=3$ potential red neighbors of the same layer $j$ in the interior of~$D'$. Again, there are at most $7$ additional potential red neighbors in layer $j$ on $D'$. The number of potential interior red neighbors may increase only {\em after} the initial contractions inside $V_j^2$ and $V_j^3$ happen, and at that moment we contract down to at most $4$ vertices in layer $j$ which means at most $3$ red neighbors there, too.
\item In layer $j+1$ we have up to $(2\cdot4+1)+7=16$ (including on~$D'$) potential red neighbors, which altogether bounds the red degree in layer $j$ along the partial sequence to at most $(4+7)+(3+7)+16=37$. Note that this number of potential red neighbors does not increase when processing layers higher than~$j$.
\end{itemize}
Let $G_5$ be the resulting trigraph of this stage.

\smallskip
\item In the second stage, we again proceed by layers $j=1,2,3,\ldots$. We now, in $G_5$, directly contract (in any order) the pairs of vertices of layer $j$ from $V_j^1$, from $V(R_2\cup R_3)\cap V_j$ and from already contracted $V_j^2\cup V_j^3$, which have the same $G_5$-neighborhood in $\{a,b\}$. The red degrees on $C$ do not increase any more (i.e., they stay at $\leq12$) and $a,b$ do not get any red edge. As for the red degrees of the interior vertices, we analyse the situation similarly to the previous point:
\vspace*{-1ex}%
\begin{itemize}
\item When processing layer $j$, there are at most $4$ vertices in layer $j-1$ in the interior of~$C$, and at most $5$ in layer $j-1$ on $C$ (one from every vertical path $P_1,\ldots,P_k$).
\item In the layers $j$ and $j+1$, we start with at most $(2\cdot4+2)+5=15$ (including on~$C$) vertices which are potential red neighbors, but in layer $j$ we are just contracting a pair which gives at most $15-1-1=13$ potential neighbors to a vertex of layer $j$.
\item Altogether, we can bound the red degree in layer $j$ along this partial sequence to at most $(4+5)+(15-2)+15=37$.
\end{itemize}
\end{enumerate}
We are left with the second case of $a\in D_2$ and $b\in D_3$, which is analysed quite similarly to the previous case of $a,b\in D_1$, and so we only sketch the arguments:
\vspace*{-1ex}%
\begin{enumerate}
\item In the first stage, we again proceed layer-by-layer with $j=1,2,3,\ldots$, and again first contract inside $V_j^2$ and $V_j^3$, and then together with $V(R_1)\cap V_j$, according to the same $G_4$-neighborhood in $\{a_2=a,\>b_3=b\}$. This way $r_1$ gets up to $4$ red edges from the layer $j$, and no other vertex on $D'$ gets new red edges.

For the vertices in the interior of $D'$, when processing layer $j$, we again have at most $4+7=11$ potential red neighbors in layer $j-1$, and initially at most $3+7=10$ such in layer $j$. And again, when we start contracting vertices across $V_j^2$ and $V_j^3$, and with $V(R_1)\cap V_j$, we have at most $3+7=10$ potential red neighbors in layer $j$. Together with up to $16$ potential red neighbors in layer $j+1$ (including on~$D'$), we get red degree~$\leq 11+10+16=37$.

\smallskip
\item In the second stage, we again proceed by layers $j=1,2,3,\ldots$, and directly contract current layer $j$ in the interior of $C$ according to the same $G_5$-neighborhood in~$\{a,b\}$. No new red edge is added to the vertices on $C$, and $a,b$ stay without any red edge. The count for potential red neighbors in layers $j-1$, $j$ and $j+1$ again gives the same numbers for the red degree of vertices in the interior of $C$, namely at most~$9+13+15=37$.
\end{enumerate}
We are done with \cref{cl:twinrecur}, and hence with the proof of the whole theorem.
\end{proof}

\noindent The fact that, unlike~\cite{DBLP:journals/corr/abs-2201-09749}, we do not exploit the planarity of $G$ to obtain the contraction sequence, allows us to get the following extension of \cref{cref:ttwin-width-planar}, by adopting into the previous proof the decomposition technique of \cref{lem:recurse} in place of the one of~\cite{DBLP:journals/corr/abs-2108-00198}.

\begin{theorem}\label{cref:ttwin-width-hframed}
Let $G$ be a simple spanning subgraph of an $h$-framed graph with~$h\geq4$.
Then the twin-width of $G$ is at most $33\floor{h/2}+\floor{h/3}+13\leq 17h+13$.
\end{theorem}

\begin{proof}
We follow the main ideas of the proof of \cref{cref:ttwin-width-planar}, and adapt them to the recursive decomposition from \cref{lem:recurse}. Let $G^+\supseteq G$ be an $h$-framed (topological) graph on the same vertex set. As in the proof of \cref{thm:prodstruct-h}, we fix the outer face of $\sk(G^+)$ bounded by a cycle $S$, such that~$3\leq|V(S)|\leq h$, and choose~$t\in V(S)$. We initially partition the vertex set of $S$ into three disjoint paths; the path $S_1$ only consisting of the vertex $t$, and two paths $S_2$ and $S_3$ of lengths at most $\floor{\frac{h-1}2}$ and $\floor{\frac{h}2}$ covering the rest of~$S$. We consider the {BFS layering} $\ca L=(V_0,V_1,\ldots)$ of $\sk(G^+)$ (where $V_i$ contains all vertices at distance $i$ from $t$). We will use the following special property of paths of $\sk(G^+)$ (with respect to fixed $\ca L$):
\begin{enumerate}[label={(V)}, ref=(V)]
\item A path $P\subseteq\sk(G^+)$ is \emph{near-vertical} if, for any $j\geq0$, the union 
$V_j\cup V_{j+1}\cup\ldots\cup V_{j+2\floor{h/2}}$ of $2\floor{h/2}+1$ consecutive layers of $\ca L$ intersects $V(P)$ in at most $5\floor{h/2}+1$ vertices.
\label{it:propertyB}
\end{enumerate}
Note that each of the paths $S_1,S_2,S_3$ decomposing $S$ trivially satisfies \cref{it:propertyB}. We have the following natural analogy of \cref{prop:valid-layering} for an $h$-framed graph $G^+$ which also extends to $G$ since every edge of $G$ has its ends at distance $\leq\floor{h/2}$ in $\sk(G^+)$:

\begin{claim}\label{cl:twwlayersh}
Let $G'$ be a trigraph obtained from $G$ by an $\ca L$-respecting partial contraction sequence.
If there is an edge (in particular, a red edge) in $G'$ from $V_j[G']$ to $V_{j'}[G']$, then~$|j-j'|\leq\floor{h/2}$.
Consequently, a vertex of $G'$ may have its red neighbors in at most $2\floor{h/2}+1\leq h+1$ layers of $\ca L$ (including its own layer).
\end{claim}

\noindent The core of our proof is the following recursive claim:

\begin{claim}\label{cl:twinrecurh}
Let $C$ be a cycle in $\sk(G^+)$ such that $t$ is not in the interior of $C$, let $G^+_C$ be the subgraph of $G^+$ bounded by~$C$, and $G_C:=G^+_C\cap G$. Let $a,b\in V(C)$ be arbitrary. Let $U:=V(G_C)\setminus V(C)$. Assume that, for some $1\leq k\leq 3$, the vertex set $V(C)$ is partitioned into $k$ pairwise disjoint near-vertical paths $P_1,\ldots,P_k\subseteq C$ (\cref{it:propertyB}). Then there exists an $\ca L$-respecting partial contraction sequence of $G$ which contracts only vertex pairs in $U$, results in a trigraph $G_0$, and satisfies the following:
\vspace*{-1ex}%
\begin{enumerate}[label={\alph*)}, ref=(\alph*)]
\item \label{it:redboundaryh} The red degree of each vertex of $V(C)$ is at most $8\floor{h/2}+8$ along the sequence
and at most $8\floor{h/2}+4$ in final~$G_0$, and no red edge (ever) is incident to either of~$a,b$;

\item \label{it:redinsideh} the red degree of vertices stemming from $U$ along the sequence is at most $33\floor{h/2}+\floor{h/3}+13$; and

\item no two vertices of $V_{j}^0:=(V(G_0)\setminus V(C))\cap V_j[G_0]$, for $j=1,2,\ldots$, have the same $G_0$-neighborhood in~$\{a,b\}$. So, in particular, $|V_j^0|\leq4$.\label{it:twinrecurhc}
\end{enumerate}
\end{claim}

Having this at hand, we can easily finish the proof of \cref{cref:ttwin-width-hframed}: We apply \cref{cl:twinrecurh} to~$C:=S$, and in the resulting trigraph $G_0$ we (as in the proof of \cref{cref:ttwin-width-planar}) contract $V_0[G_0]=\{t\}$ with $V_1[G_0]$ into one vertex, then with $V_2[G_0]$, with $V_3[G_0]$, and so on.

By \cref{cl:twwlayersh}, we have that only $(h+1)$ layers of $\ca L$ contribute to the maximum red degree, which is thus at most $4(h+1)+h-1=5h+3< 33\floor{h/2}$ ($+h$ simply estimates the contribution of the whole cycle~$S$) in this remaining contraction sequence of~$G_0$.

So, the task is to prove \cref{cl:twinrecurh}, which we do by induction as in the proof of \cref{lem:recurse}, and using analogous observations. In particular, although our graph $G$ is not planar in general, its edges cannot cross the cycle $C$ of $\sk(G^+)$ by the definition of a skeleton.

From the inductive proof of \cref{lem:recurse}, we get a subgraph $R\subseteq\sk(G^+)\cap G_C$ such that $X:=V(R)\setminus V(C)$ satisfies \Cref{cl:inclusion} of \cref{lem:recurse}. Considering the $a\geq2$ bounded faces of $C\cup R$ (which is $2$-connected), we denote their bounding cycles by $D_1,\ldots,D_a$, and say that the cycle $D_i$ is \emph{nonempty} if the interior of $D_i$ contains a vertex of $G_C$. We furthermore have (from the proof) that for each $i\in\{1,\ldots,a\}$, the intersection $D_i\cap C$ is a path hitting at most two of the paths $P_1,\ldots,P_k$ partitioning~$C$, and that every vertex of $C$ (of $R-V(C)$) belongs to at most two (three, respectively) nonempty ones of $D_1,\ldots,D_a$. Note that $a$ may be higher than $3$ (cf. the separable case of \cref{lem:recurse}). We claim that the paths $Q_i:=(R\cap D_i)-V(C)$, $i\in\{1,\ldots,a\}$, satisfy \cref{it:propertyB}. Since $V(Q_i)\subseteq V(R)\setminus V(C)=X$, \cref{it:propertyB} is directly implied by \Cref{cl:inclusion} of \cref{lem:recurse}, unless $q=3$ in the latter -- then, the proof of \cref{lem:recurse} uses three vertical paths of $\sk(G^+)$ to make $R$, and only at most two of them belong to $D_i$, while the third one may possibly intersect $D_i$ in one vertex (an end). Hence, a union of $2\floor{h/2}+1$ consecutive layers of $\ca L$ intersects $Q_i$ in at most $2(2\floor{h/2}+1)$ vertices of the two vertical paths plus additional (cf.~$|X'|$ in \Cref{cl:inclusion}) at most $|X'|+1\leq(\floor{h/3}-1)+1\leq \floor{h/2}-1$ vertices, which altogether gives $\leq5\floor{h/2}+1$ as desired.

We are going to apply \cref{cl:twinrecurh} independently to each nonempty one of the cycles $D_i$, $i\in\{1,\ldots,a\}$. For that, we first determine the special pair $a_i,b_i\in V(D_i)$ anticipated in the claim. If $V(D_i)\cap\{a,b\}=\emptyset$, then we choose $a_i$ and $b_i$ as the ends of the path $D_i\cap C$, and we take one or both of $a_i,b_i$ arbitrarily if $D_i\cap C$ is one-vertex or empty. If $a,b\in V(D_i)$, then we choose $a_i=a$ and $b_i=b$. If, up to symmetry, $a\in V(D_i)$ and $b\not\in V(D_i)$, then we set $a_i=a$ and $b_i$ as one of the ends of $D_i\cap C$ -- we take an arbitrary of the two ends except the special case described next; if $a\in V(D_i)$, $b\in V(D_{i'})$ and one end (or both ends) of $D_i\cap C$ is an end (ends) of $D_{i'}\cap C$, then we set $b_i$ to the end shared with $D_{i'}\cap C$ (and the possible other shared end is set to $b_{i'}$). The applications of \cref{cl:twinrecurh} to nonempty ones of the cycles $D_i$, $i\in\{1,\ldots,a\}$, result in partial contraction sequences of $G$ which compose (independently one after another) into a partial contraction sequence from $G$ to a trigraph~$G_1$.

Let $W_i$ denote the set of vertices of $G_1$ that stem from the vertices in $V(G_{D_i})\setminus V(D_i)\subseteq U$ in the previous contractions. Note that the vertices $a$ and $b$ have no red edges in $G_1$, and that every other vertex of $C$ has red degree as claimed by \cref{it:redboundaryh} of \cref{cl:twinrecurh} along the partial sequence to~$G_1$; the latter holds also for the vertices $x$ of $C$ which belong to two nonempty cycles of $D_1,\ldots,D_a$, because $x$ has been set as $a_i$ or $b_i$ for at least one of them (and hence received no red edge). A vertex $z\in V(R)\setminus V(C)$ may belong to at most three nonempty of the cycles $D_1,\ldots,D_a$, and so by (again) \cref{it:redboundaryh} of \cref{cl:twinrecurh}, its red degree along the partial sequence to~$G_1$ is always at most $3(8\floor{h/2}+8)\leq33\floor{h/2}+6$. No other vertex is supplied with red edges from more than one of the recursive applications, and so the partial contraction sequence from $G$ to $G_1$ fulfills the red-degree conditions of \cref{cl:twinrecurh}. Assume, up to symmetry, that $a,b\in V(D_1)\cup V(D_2)$. Set $W_1'=W_1$. For $i:=2,\ldots,a$, we obtain a trigraph $G_i$ from $G_{i-1}$ by the following partial contraction sequence:
\begin{itemize}
\vspace*{-1ex}%
\item Let $X_i=V(R)\setminus(V(D_{i+1})\cup\ldots V(D_a))$ and $Y_i=W'_{i-1}\cup W_i\cup(X_i\cap U)$.
\item Iteratively, for $j:=1,2,3,\ldots\,$, and over the sets $N:=\emptyset,\{a\},\{b\},\{a,b\}$, we contract into one vertex the set $Z_{i,j}^N$
	where $Z_{i,j}^N\subseteq Y_i$ is formed by those vertices which are in layer $j$ of $\ca L$ and have the same $G_{i-1}$-neighborhood $N$ in $\{a,b\}$.
	Precisely, we pick a vertex of $Z_{i,j}^N$ at random, and successively contract other vertices of $Z_{i,j}^N$ with it in any order.
\item We denote by $W_i'$ the vertices of $G_i$ that stem from the vertices of $Y_i$ in the previous.
\end{itemize}

Importantly, the set $Y_i$ in round $i$ captures all vertices which have participated in a contraction since $G_1$, and those which are going to be contracted in this round. The vertices $a,b$ receive no red edges during this procedure, and in every round~$i$, no red edges are created into the vertices of $W_{i'}$ for any~$i'>i$ or into~$V(G)\setminus V(G_C)$. Furthermore, observe that $W_i$ has at most $4$ vertices in every layer of $\ca L$ by \cref{it:twinrecurhc} of \cref{cl:twinrecurh}, and likewise $W_i'$ has at most $4$ vertices in every layer of $\ca L$ by the condition of distinct neighborhoods in~$\{a,b\}$.

The remaining task is to verify the claimed properties of \cref{cl:twinrecurh} for the partial contraction sequence from $G_1$ to~$G_0:=G_a$ (as described by the previous paragraph). \cref{it:twinrecurhc} is trivial -- we have just contracted the vertices that way. Regarding \cref{it:redinsideh} of \cref{cl:twinrecurh}, we estimate the red degree of every vertex $z$ in the interior of $C$ along the contraction subsequence from $G_{i-1}$ to $G_i$ (for all $i:=2,\ldots,a$ as above). First, if $z\in W_{i+1}\cup\ldots\cup W_a$, then the red degree of $z$ is as required from the recursive application of \cref{cl:twinrecurh} -- since no other red edge to it has been created so far. If $z\in V(R)\setminus Y_i$ (in particular, $z$ has not yet participated in a contraction),  then $z$ is incident to at most three nonempty cycles $D_{i'}$, $i'\in I'$, and if $\min(I')>i$, then no red edge to $z$ has been created since~$G_1$ and the red degree of $z$ is as needed. Otherwise, $z$ may have red edges to at most $2(8\floor{h/2}+4)$ vertices of the (at most two) sets $W_{i'}$ where $i'\in I'$ and $i'>i$ by \cref{it:redboundaryh} of \cref{cl:twinrecurh}, and to some vertices of $W'_{i-1}\cup W_i$ (which include all vertices contracted since~$G_1$). By \cref{cl:twwlayersh}, at most $2\floor{h/2}+1$ layers of $\ca L$ may host red neighbors of~$z$, and $W'_{i-1}\cup W_i$ has at most $4+4=8$ vertices in each layer (as observed above), summing to $8(2\floor{h/2}+1)$. Altogether, $2(8\floor{h/2}+4)+8(2\floor{h/2}+1)= 32\floor{h/2}+16\leq 33\floor{h/2}+\floor{h/3}+13$ is an upper bound on the red degree of~$z\in V(R)\setminus Y_i$.

Lastly, we have the case of $z\in Y_i$. By the definition of $Y_i$, there can be no red edge from $z$ to $W_{i+1}\cup\ldots\cup W_a$. Potential red neighbors of $z$ in $V(C)\cup V(R)$ are estimated as follows. Again, by \cref{cl:twwlayersh}, at most $2\floor{h/2}+1$ layers of $\ca L$ may host red neighbors of~$z$. Consequently, the number of such red neighbors in $V(C)$ is at most $3(5\floor{h/2}+1)$ by \cref{it:propertyB}, and in $V(R)$ it is at most  $3(2\floor{h/2}+1)+\floor{h/3}-1$ by \Cref{cl:inclusion} of \cref{lem:recurse}. Note that the previous estimate accounts also for the vertices of $V(R)$ in $Y_i$, and so the remaining red neighbors of $z$ must belong to $W'_{i-1}\cup W_i$, which we bound in the next paragraph.

Assume that we are now contracting in layer $j$ of $\ca L$ (see the above contraction procedure) and that $z$ belongs to layer $j'$. Then, by \cref{cl:twwlayersh}, red neighbors of $z$ may lie in layer $j'$ and in the $\floor{h/2}$ previous and $\floor{h/2}$ next layers. If $j'\leq j$, then in the $\floor{h/2}$ previous layers, $W'_{i-1}\cup W_i$ has already been contracted down to at most $4$ vertices per layer, and hence the number of red neighbors of $z$ in $W'_{i-1}\cup W_i$ is at most $4\floor{h/2}+8(\floor{h/2}+1)=12\floor{h/2}+8$. If $j'>j$, then $z$ has not participated yet in a contraction between $W'_{i-1}$ and $W_i$, and so $z$ has no red neighbors of layers $\geq j'$ in $W_i$ if $z\in W'_{i-1}$ (in $W'_{i-1}$ if $z\in W_i$, resp.). Therefore, this time the number of red neighbors of $z$ in $W'_{i-1}\cup W_i$ is at most $4(\floor{h/2}+1)+8\floor{h/2}=12\floor{h/2}+4$.
Summing the previous terms together, the red degree of $z\in Y_i$ is always at most $3(5\floor{h/2}+1) + 3(2\floor{h/2}+1)+\floor{h/3}-1 + 12\floor{h/2}+8 = 33\floor{h/2}+\floor{h/3}+13$.

At last, we check \cref{it:redboundaryh} of \cref{cl:twinrecurh}. This has already been verified for $G_1$ above. However, it will be useful to understand this bound on the red degrees of $V(C)$ in $G_1$ in closer detail. To recapitulate, a vertex $z\in V(C)$ can have, in $G_1$, red neighbors in at most $2\floor{h/2}+1$ layers of $\ca L$ by \cref{cl:twwlayersh}, and at the same time in vertices that belong to at most one of the sets $W_1,\ldots,W_a$ by our choice of $a_i,b_i$ in the recursive applications of \cref{cl:twinrecurh}. Furthermore, each set $W_i$ has at most $4$ vertices in one layer of $\ca L$ (see above), and this leads to the upper bound of $4(2\floor{h/2}+1)=8\floor{h/2}+4$.

We are going to argue that at every step along the contraction subsequence from $G_{i-1}$ to $G_i$ (for all $i:=2,\ldots,a$ as above), a vertex $z\in V(C)$ has at most $4$ red neighbors (in $Y_i$) in every layer of $\ca L$ except the layer $j$ which is currently being contracted (cf.~the sets $Z_{i,j}^N$ above). For layers $j'<j$, this is trivial since our contraction subsequence has left only at most $4$ vertices (by four distinct neighborhoods in $\{a,b\}$) as potential red neighbors. For layers $j'>j$ we have, by the previous arguments about $G_1$, that red neighbors of $z$ are either in $W_{i-1}'$ or in $W_i$, but not in both, and so the upper bound of $4$ again follows. In layer $j$, the order of contractions in the sets $Z_{i,j}^N$ ensure that at most one new red edge adds to $z$ for each of the four choices of $N$ (and, of course, these potentially added edges are again dismissed after $Z_{i,j}^N$ is contracted into one vertex). By a rough estimate, $z$ has at most $4+4=8$ red neighbors in layer~$j$, together at most $4\cdot2\floor{h/2}+8=8\floor{h/2}+8$, as desired. At the end of this subsequence, that is in $G_i$, the latter bound readily decreases to $8\floor{h/2}+4$. We are done with a proof of \cref{cl:twinrecurh}, and hence finished the proof of the theorem.
\end{proof}

\begin{corollary}
The twin-width of simple $1$-planar and optimal $2$-planar graphs is at most~$80$.
\end{corollary}

\noindent We point out that \cref{cref:ttwin-width-hframed} implies an improvement on the twin-width of $k$-map~graphs only up to a certain $k$, as these graphs have bounded twin-width~independently~of~$k$~\cite{DBLP:journals/jacm/BonnetKTW22}.

\section{Conclusions}
\label{sec:conclusions}

Our structural results are constructive and can easily be implemented to run in quadratic time, provided that the input graph is a topological $h$-framed graph. A major open question is to obtain a speed up in these constructions. The recent algorithmic advances in~\cite{DBLP:journals/corr/abs-2202-08870,DBLP:journals/corr/abs-2004-02530} may lead to linear-time implementations. Another important open problem is whether each $k$-planar graph is a subgraph of the strong product of a path, a (planar) graph of constant treewidth, and a clique of size linear in $k$, as our results suggest that such a structure might be possible.

\bibliographystyle{plainurl}
\bibliography{general,stacks,queues,tww}

\end{document}